\documentclass[UKenglish]{llncs}
\usepackage{amssymb,amsmath}
\newtheorem{claimx}{{\bf Claim}}
\newtheorem{observation}{Observation}
\usepackage{wrapfig}
\usepackage{xspace}
\usepackage{xcolor}
\usepackage{paralist}
\usepackage{hyperref}
\hypersetup{pdfborder={0 0 0.4}}
\usepackage{todonotes}
\usepackage{graphicx}
\usepackage[font=small,labelfont=bf]{caption}
\usepackage[subrefformat=parens, labelfont=default]{subcaption}
\usepackage{thm-restate}
\usepackage{cleveref}
\usepackage{cite}
\usepackage[basic]{complexity}
\newif\ifarxive
\arxivetrue % or
\ifarxive
\newcommand{\springerarxiv}[2]{{#2}\xspace}
\else
\newcommand{\springerarxiv}[2]{{#1}\xspace}
\fi

\spnewtheorem*{sketch}{Sketch of proof}{\itshape}{\rmfamily}

\title{Planar L-Drawings of Directed Graphs\thanks{\ack}}

\newcommand{\ack}{This research was initiated at
  the Bertinoro Workshop on Graph Drawing 2017. 
  This article reports
  on work supported by the U.S.~Defense Advanced Research Projects
  Agency (DARPA) under agreement no.~AFRL FA8750-15-2-0092. The views
  expressed are those of the authors and do not reflect the official
  policy or position of the Department of Defense or the U.S.\
  Government.  This research was also partially supported by MIUR
  project ``MODE -- MOrphing graph Drawings Efficiently'',
  prot.~20157EFM5C\_001.
  }

\authorrunning{Chaplick~\etal}

\author{Steven Chaplick \inst1 \and Markus Chimani \inst2 \and Sabine
  Cornelsen \inst3 \and Giordano~Da~Lozzo \inst4 \and
  Martin~N\"ollenburg \inst5 \and Maurizio~Patrignani \inst6 \and
  Ioannis~G.~Tollis \inst7 \and Alexander~Wolff \inst8}
  
\institute{Universit\"at W\"urzburg, Germany;
\email{steven.chaplick@uni-wuerzburg.de} 
\and Universit\"at Osnabr\"uck, Germany;
\email{markus.chimani@uni-osnabrueck.de}
\and Universit\"at Konstanz, Germany;
\email{sabine.cornelsen@uni-konstanz.de}
\and University of California, Irvine, CA, USA;
\email{gdalozzo@uci.edu}
\and TU Wien, Austria; 
\email{noellenburg@ac.tuwien.ac.at}
\and Roma Tre University, Rome, Italy; 
\email{patrigna@dia.uniroma3.it} 
\and University of Crete, Heraklion, Greece;
\email{tollis@csd.uoc.gr}
\and
Universit\"at W\"urzburg, Germany; 
\href{http://orcid.org/0000-0001-5872-718X}{\color{black}\tt
  orcid.org/0000-0001-5872-718X}
}

\newcommand{\st}{st}
\def\etal{{\em et~al.}}

\def\eg{{e.g.}}
\def\Tee{{\texttt{T}}}
\def\Plus{{\texttt{+}}}
\def\Ell{{\texttt{L}}}

\graphicspath{{figures/}}

\definecolor{blue}{rgb}{0.274,0.392,0.666}
\definecolor{ourred}{rgb}{1,0.3,0.3}
\definecolor{ourgreen}{rgb}{0,0.588,0.509}

\newcommand{\red}[1]{{#1\xspace}}

\newcommand{\remove}[1]{}%{{#1}\xspace}

\DeclareMathOperator{\pert}{pert}
\DeclareMathOperator{\skel}{skel}

\let\doendproof\endproof
\renewcommand\endproof{\hfill $\qed$\doendproof}

\let\doendsketch\endsketch
\renewcommand\endsketch{\hfill $\qed$\doendsketch}

\begin{document}

\maketitle
\begin{abstract}
We study planar drawings of directed graphs in the L-drawing standard. We provide necessary conditions for the existence of these drawings and show that testing for the existence of a planar L-drawing is an \NP-complete problem. Motivated by this result, we focus on upward-planar L-drawings.
% We study upward-planar drawings of \st-graphs in the L-drawing standard. 
We show that directed \st-graphs
admitting an upward- (resp.\ upward-rightward-) planar L-drawing are exactly those admitting a bitonic (resp.\ monotonically increasing) \st-ordering. 
%and that directed \st-graphs admitting an upward-rightward-planar L-drawing are exactly those admitting a monotonically-decreasing \st-ordering. 
We give a linear-time algorithm that computes a bitonic (resp.\ monotonically increasing) \st-ordering of a planar \st-graph or reports that there exists none. 
\end{abstract}

\section{Introduction}
%~\cite{angelini_etal:sofsem16}
In an \emph{L-drawing} of a directed graph each vertex~$v$ is assigned a point in the plane with exclusive integer $x$- and $y$-coordinates, and
 each directed edge $(u,v)$ consists of a vertical segment exiting~$u$ and of a horizontal segment entering~$v$~\cite{angelini_etal:sofsem16}. 
The drawings of two edges may cross and partially overlap, 
%  in either their
% horizontal or vertical segments, 
following the model of~\cite{kt-ood-11}.
The ambiguity among crossings and bends is resolved by replacing bends with small rounded junctions.
An L-drawing in which edges possibly overlap, but do not cross, is a {\em planar L-drawing};
% of a directed, planar graph and the corresponding L-drawing, 
see, e.g., Fig.~\ref{fig:example}\red{b}.
A planar L-drawing is {\em upward planar} if its edges are $y$-monotone, and it is {\em upward-rightward planar} if its edges are simultaneously $x$-monotone and $y$-monotone. 
\begin{figure}[tb]
  \centering
  \begin{tabular}[b]{@{}c@{~~}c@{~~}c@{~~}c@{~~}c@{}}
    (a)\hspace{-2ex}\includegraphics[page=4,width=.2\textwidth]{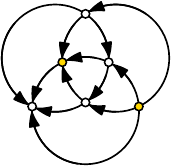} &
    (b)~{\raisebox{7.5em}{\includegraphics[page=5,width=.2\textwidth,angle=180]{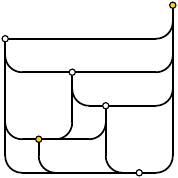}}} &
    (c)~{\raisebox{7.5em}{\includegraphics[page=6,width=.2\textwidth,angle=180]{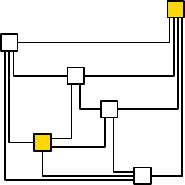}}} &
    (d)~\includegraphics[page=3,width=.2\textwidth]{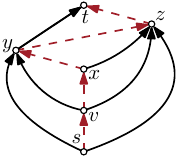}
  \end{tabular}

  \caption{%
    (a) A bitonic st-orientation of the octahedron that admits an
    upward planar L-drawing~(b). (c)~The corresponding drawing in the
    Kandinsky model.
    (d) An upward planar \st-graph
    $U$ that does not admit an upward-planar L-drawing 
    %, and a contact representation by T-shapes~(d).
  }
  \label{fig:example}
\end{figure}

%{\bf Relationship with the Kandinsky Model.}
Planar L-drawings correspond to drawings in the
Kandinsky model~\cite{foessmeier/kaufmann:gd95} with exactly one bend
per edge and with some restrictions on the angles around each vertex;
see Fig.~\ref{fig:example}\red{c}.
It is
\NP-complete~\cite{blaesius_etal:esa14} to decide
whether a multigraph has a planar embedding that allows a Kandinsky
drawing with at most one bend per edge~\cite{brueckner:ba13}. 
% \todo{Gio: here is a 1st  parallelism: We are basically proving that it is NP-complete to decide whether a simple planar graph admits a Kandinsky drawing with {\bf exactly} one bend per edge.}
On the other hand, every simple planar graph has a Kandinsky drawing with at most one bend per edge~\cite{brueckner:ba13}. 
% \todo{Gio: here is a 2nd  parallelism: we prove/know that every simple planar graph has a planar L-drawing drawing with {\bf exactly} one bend per edge.}
Bend-minimization in the Kandinsky-model is
\NP-complete~\cite{blaesius_etal:esa14} even if a
planar embedding is given, but can be approximated by a factor of
two~\cite{eigelsperger:phd,barth_etal:gd06}.
Heuristics for drawings in the Kandinsky model with empty faces and
few bends have been discussed by Bekos et al.~\cite{bekos_etal:sea15}.

%{\bf Bitonic \st-orderings.}  
Bitonic \st-orderings were introduced by
Gronemann for undirected planar graphs~\cite{gronemann:gd14} as an
alternative to canonical orderings. They were recently extended to
directed plane graphs~\cite{g-bsupg-GD16}. In a bitonic \st-ordering
the successors of any vertex must form an increasing and then a
decreasing sequence in the given embedding.
More precisely,
a {\em planar \st-graph} is a directed acyclic graph with a single
source $s$ and a single sink $t$ that admits a planar embedding in
which $s$ and $t$ lie on the boundary of the same face. A planar \st-graph always admits an upward-planar straight-line drawing~\cite{DT-aprad-88}.
An \emph{\st-ordering} of a planar \st-graph is an enumeration $\pi$
of the vertices with distinct integers, such that $\pi(u) < \pi(v)$
for every edge $(u,v) \in E$.  Given a {\em plane \st-graph}, i.e., a
planar \st-graph with a fixed upward-planar embedding $\cal E$,
consider the list $S(v)=\left< v_1,v_2,\dots,v_k \right>$ of
successors of $v$ in the left-to-right order in which they appear
around $v$.
The list $S(v)$ is \emph{monotonically decreasing} with respect to an \st-ordering
$\pi$ if $\pi(v_i)>\pi(v_{i+1})$ for $i=1,\dots,k-1$. It is {\em bitonic}
with respect to $\pi$ if there is a vertex $v_h$ in $S(v)$ such that
$\pi(v_i)<\pi(v_{i+1})$, $i=1,\dots,h-1$ and $\pi(v_i)>\pi(v_{i+1})$,
$i=h,\dots,k-1$.
For an upward-planar embedding $\cal E$, an \st-ordering $\pi$ is {\em bitonic} or {\em monotonically decreasing}, respectively if the successor list of each vertex is bitonic or monotonically decreasing, respectively.
Here, $\left< {\cal E}, \pi \right>$ is called a {\em bitonic pair} or {\em monotonically decreasing pair}, respectively, of~$G$. 
 
Gronemann used bitonic \st-orderings to 
%extended bitonic \st-orderings
%to directed plane graphs.  To this end, he 
obtain on the one hand upward-planar
polyline grid drawings in quadratic area with at most $|V|-3$ bends in
total~\cite{g-bsupg-GD16}  and on the other hand contact representations with upside-down oriented T-shapes~\cite{gronemann2015algorithms}.
A bitonic \st-ordering for biconnected undirected planar graphs can be computed in linear time~\cite{gronemann:gd14} and the existence of a bitonic \st-ordering for plane (directed) \st-graphs can also be decided in linear time~\cite{g-bsupg-GD16}. However, in the variable embedding scenario no algorithm is known to decide whether an \st-graph $G$ admits a bitonic pair.
Bitonic \st-orderings turn out to be strongly related to upward-planar L-drawings of \st-graphs. In fact, the $y$-coordinates of an upward-planar L-drawing yield a bitonic \st-ordering. 

%\paragraph{\bf Our Contribution.}
In this work, we initiate the investigation of planar and upward-planar L-drawings. 
%and contribute to shedding light on the properties of bitonic \st-orderings. 
In particular, our contributions are as follows.
\begin{inparaenum}[(i)]
\item We prove that deciding whether a directed planar graph admits a planar L-drawing is \NP-complete.
%%%%%%%%%%%%%%%%%%%%%%%%%%%%%%%%%%%%%%%%%%%%%%%%%%%%%%%%%%%%%%%%%%%%%%%
%Sab: This repeated Theorem 1 over and over again
%%%%%%%%%%%%%%%%%%%%%%%%%%%%%%%%%%%%%%%%%%%%%%%%%%%%%%%%%%%%%%%%%%%%%%%
 \item We characterize the planar \st-graphs admitting an upward
   (upward-rightward, resp.) planar L-drawing as the \st-graphs
   admitting a bitonic (monotonic decreasing, resp.) \st-ordering.  
\item We provide a linear-time algorithm to compute an embedding, if any, of a planar \st-graph that allows for a bitonic \st-ordering. This result complements the analogous algorithm proposed by Gronemann for undirected graphs~\cite{gronemann:gd14} and extends the algorithm proposed by Gronemann for planar \st-graphs in the fixed embedding setting~\cite{g-bsupg-GD16}.
\item
Finally, we show how to decide efficiently whether there is a planar L-drawing for a plane directed graph with a fixed assignment of the edges to the four ports of the vertices. %To this end, we express planar L-drawings by linear inequalities that basically describe a network flow plus the bend-or-end property of the Kandinsky model.
\end{inparaenum}

\springerarxiv{Due to space limitations, full proofs are provided in~\cite{cccdnptw-plddg-tr-17}.}{Due to space limitations, full proofs are provided in~\ref{apx:proofs}.}

\section{Preliminaries}

We assume familiarity with basic graph drawing concepts and in particular with the notions of connectivity and SPQR-trees \springerarxiv{(see also~\cite{dt-ogasp-90,cccdnptw-plddg-tr-17}).}{(see also~\cite{dt-ogasp-90} and~\ref{apx:spqr}).}  

A (simple, finite) directed graph $G=(V,E)$ consists of a finite set $V$ of vertices
and a finite set $E \subseteq \{(u,v) \in V \times V; u \neq v\}$ of ordered pairs of vertices. If $(u,v)$ is an edge then $v$ is a
\emph{successor} of $u$ and $u$ is a \emph{predecessor} of
$v$. 
%
%%%%%%%%%%%%%%%%%%%%%%%%%%%%%%%%%%%%%%%%%%%%%%%%%%%%%%%%%%%%%%%%%%%%%%%%%%%
A graph is \emph{planar} if it admits a 
drawing in the plane without edge crossings.
A \emph{plane graph} is a planar graph with a fixed \emph{planar embedding}, i.e., with fixed circular orderings of the edges incident to each vertex---determined by a planar drawing---and with a fixed outer face.

Given a planar embedding and a vertex $v$, a pair of consecutive edges incident to $v$ is {\em alternating} if they are not both incoming or both outgoing. We say that $v$ is \emph{$k$-modal} if there exist exactly $k$ alternating pairs of edges in the cyclic order around $v$.
%
% %
An embedding of a directed graph $G$ is  \emph{$k$-modal}, %with $k \geq 1$, 
if each vertex is at most $k$-modal.
%if the edges incident to each vertex form at most $2k$ sets of homogeneous incoming or outgoing edges, and sets of incoming and outgoing edges alternate in the cyclic ordering of the edges at the vertex; 
A $2$-modal embedding is also called \emph{bimodal}. An upward-planar drawing determines a bimodal embedding. However, the existence of a bimodal embedding is not a sufficient condition for the existence of an upward-planar drawing. Deciding whether a directed graph admits an upward-planar (straight-line) drawing is an \NP-hard problem~\cite{GargTamassia01}.

\paragraph{\bf L-drawings.}
\begin{wrapfigure}[6]{hR}{.2\textwidth}
    \centering
    \vspace{-2.1em}
    \includegraphics[page=4,width=.2\textwidth]{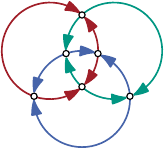}
    % \caption{Furthest Bend.}\label{FIG:kandinskyFurthestBend}
\end{wrapfigure}
% 

%Observe that, 
A planar L-drawing determines a $4$-modal
embedding. This implies that there exist planar directed graphs that
do not admit planar L-drawings. 
%Fig.~\ref{fig:6modal} provides an
A $6$-wheel whose central vertex is incident to alternating incoming and outgoing edges is an example of a graph
% A wheel
% example of a graph 
that does not admit any $4$-modal embedding, and therefore
any planar L-drawing.
%none of whose combinatorial embeddings is~$4$-modal; 
%thus, such a graph does not admit a planar L-drawing.

On the other hand, the existence of a $4$-modal embedding is not
sufficient for the existence of a planar L-drawing.
 E.g., the
octahedron depicted in the figure on the right %in Fig.~\ref{fig:example}\red{a} 
does not admit a planar L-drawing.
 Since the octahedron is triconnected, it
admits a unique combinatorial embedding (up to a flip). Each vertex is
$4$-modal. However, the rightmost vertex in a planar L-drawing 
%cannot be left to the right and thus cannot be $4$-modal.
must be $1$-modal or $2$-modal.

%It is not difficult to see that an 
Any upward-planar L-drawing of an \st-graph $G$ can be modified to obtain an upward-planar drawing of $G$: 
%This can be done by 
Redraw each edge as a $y$-monotone curve arbitrarily close to the drawing of the corresponding $1$-bend orthogonal polyline while avoiding 
%introducing 
crossings and edge-edge overlaps. However, not every upward-planar graph admits an upward-planar L-drawing. E.g., the graph in Fig.~\ref{fig:example}\red{d} contains a subgraph that does not admit a bitonic \st-ordering~\cite{g-bsupg-GD16}.
In Section~\ref{se:fixed-embedding} (Theorem~\ref{th:characterization}), we show that this means it does not admit an upward planar L-drawing. 

\paragraph{\bf{The Kandinsky Model.}}

In the Kandinsky model~\cite{foessmeier/kaufmann:gd95}, vertices are
drawn as squares of equal sizes on a grid and edges---usually undirected---are drawn as
orthogonal polylines on a finer grid; see Fig.~\ref{fig:example}\red{c}.
Two consecutive edges in the clockwise order around a vertex
define a face and an angle in $\{0, \pi/2, \pi, 3\pi/2, 2\pi\}$ in
that face.  In order to avoid edges running through other vertices,
the Kandinsky model requires the so called \emph{bend-or-end
  property}: There is an assignment of bends to vertices with the
following three properties. (a) Each bend is assigned to at most one
vertex. (b) A bend may only be assigned to a vertex to which it is
connected by a segment (i.e., it must be the first bend on an
edge). (c) If $e_1,e_2$ are two consecutive edges in the clockwise
order around a vertex $v$ that form a 0 angle inside face $f$, then
a bend of $e_1$ or $e_2$ forming a $3\pi/2$ angle inside $f$ must be assigned to $v$. Further, the
Kandinsky model requires that there are no \emph{empty faces}.

%%%%%%%%%%%%%%%%%%%%%%%%%%%%%%%%%%%%%%%%%%%%%%%%%%%%%%%%%%%%%%%%%%%%%%
%%%%%% include the following paragraph only in the short version %%%%%
%%%%%%%%%%%%%%%%%%%%%%%%%%%%%%%%%%%%%%%%%%%%%%%%%%%%%%%%%%%%%%%%%%%%%%
Given a planar L-drawing, consider a vertex~$v$ and all edges incident
to one of the four ports of~$v$.  By assigning to~$v$ all bends on
these edges---except the bend furthest from~$v$---we satisfy the
bend-or-end property.  This implies the following lemma, which is
proven in~\springerarxiv{\cite{cccdnptw-plddg-tr-17}.}{\ref{apx:proofs}.}
%%%%%%%%%%%%%%%%%%%%%%%%%%%%%%%%%%%%%%%%%%%%%%%%%%%%%%%%%%%%%%%%%%%%%%
\begin{lemma}\label{LEMMA:kandinsky}
  A graph has a planar L-drawing if and only if it admits a drawing in
  the Kandinsky model with the following properties:
  \begin{inparaenum}[(i)]
  \item
    Each edge bends exactly once;
  \item
    at each vertex, the angle between any two outgoing (or between any two incoming) edges is 0 or $\pi$; and 
  \item
    at each vertex, the angle between any incoming edge and any outgoing edge is $\pi/2$ or~$3\pi/2$.
  \end{inparaenum}
\end{lemma}

%%%%%%%
%%%%%%%
%%%%%%%
%%%%%%%
\section{General Planar L-Drawings}
\label{se:variable}

We consider the problem of deciding whether a graph admits a planar L-drawing. 
%computing crossing-free L-drawings. 
In Section~\ref{sse:undirected}, we show that the problem is
\NP-complete if no planar embedding is given. In the fixed embedding setting (Section~\ref{SEC:ports}) the problem can be described as an ILP. It is solvable in linear time if we also fix the ports.

\subsection{Variable Embedding Setting}
\label{sse:undirected}

As a central building block for our hardness reduction we use a directed graph $W$ that can be constructed starting from a $4$-wheel with central vertex $c$ and rim $(u,v,w,z)$. We orient the edges of $W$ so that $v$ and $z$ (the \emph{V-ports} of $W$) are sinks and $u$ and $w$ (the \emph{H-ports} of $W$) are sources. Finally, we add directed edges $(v,c)$, $(z,c)$, $(c,w)$, and $(c,u)$; see Fig.~\ref{fig:4wheel}. 
We now provide Lemma~\ref{cl:rect} which describes the key property of planar L-drawings of $W$.

\begin{restatable}{lemma}{lemmawheel}
\label{cl:rect}
  In any planar L-drawing of $W$ with cycle $(u,v,w,z)$ as the outer
  face the edges of the outer face form a rectangle (that contains
    vertex $c$).
\end{restatable}

  %\begin{wrapfigure}[11]{R}{.30\textwidth}
  \begin{figure}[tb]
    \centering
    \includegraphics[page=2]{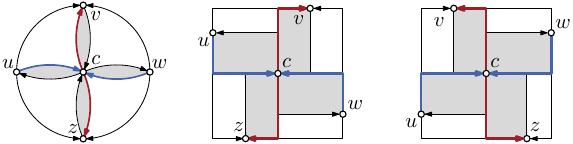}
    \caption{$4$-wheel graph $W$ and two planar L-drawings of $W$.}
    \label{fig:4wheel}
  \end{figure}
  % \end{wrapfigure}

We are now ready to give the main result of the section. 

\begin{restatable}{theorem}{theoremhardness}\label{th:hardness}
  It is NP-complete to decide whether a directed graph admits a planar
  L-drawing.
\end{restatable}

\begin{sketch}
We reduce from the \NP-complete problem of HV-rectilinear planarity testing~\cite{dlp-chvrp-14}. In this problem, the input is a biconnected degree-$4$ planar graph~$G$ with edges labeled either H or V, and the goal is to decide whether~$G$ admits an \emph{HV-drawing}, i.e., a planar drawing such that each H-edge (V-edge) is drawn as a horizontal (vertical) segment.
Starting from $G$, we construct a graph $G'$ by replacing: (i) vertices with $4$-wheels as in Fig.~\ref{fig:4wheel}; (ii) V-edges with the gadget shown in Fig.~\ref{fig:V-edge}; and (iii) H-edges with an appropriately rotated and re-oriented version of the V-edge gadget. 
%For a V-edge $(u,v)$ the two vertices of the edge gadget labeled $u$ and $v$ are identified with a V-port of the respective vertex gadgets and for an H-edge with an H-port of the vertex gadgets.
%The two vertices labeled $u$ and $v$ of the gadget for edge $(u,v)$ are identified with a V-port of the respective vertex gadgets if $(u,v)$ is a V-edge, with an H-port otherwise.
If $(u,v)$ is a V-edge, the two vertices labeled $u$ and $v$ of its gadget are identified with a V-port of the respective vertex gadgets. Otherwise, they are identified with an H-port. Figure~\ref{fig:HV-to-L} shows a vertex gadget with four incident edges.
%
%The proof concludes by showing that $G'$ admits a planar L-drawing if and only if $G$ has an HV-drawing. This part of the proof is somewhat similar to Br\"uckner's hardness proof in~\cite[Theorem~3]{brueckner:ba13}.%\qed
The proof that $G'$ and $G$ are equivalent is somewhat similar to Br\"uckner's hardness proof in~\cite[Theorem~3]{brueckner:ba13} and exploits \autoref{cl:rect}. Refer to~\springerarxiv{\cite{cccdnptw-plddg-tr-17}}{\ref{apx:proofs}} for the full details.
\end{sketch}

  \begin{figure}[tbp]
    \centering
  \begin{subfigure}[b]{.4\textwidth}
    \centering
    \includegraphics[scale=1,page=1,height=.7\textwidth]{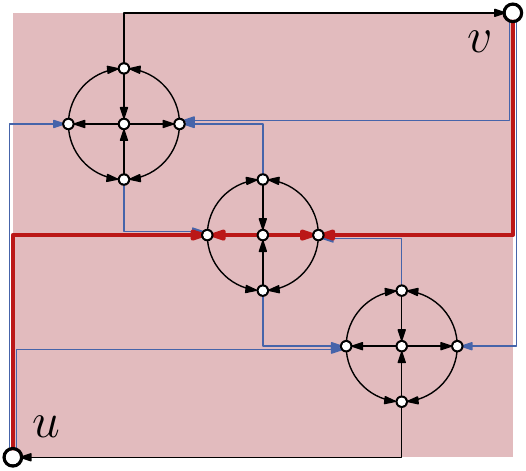}
    \caption{}
    \label{fig:V-edge}
  \end{subfigure}
\hfil
  \begin{subfigure}[b]{.4\textwidth}
        \centering
  %   \includegraphics[scale=1,page=2]{V-gadget}
  %   \caption{}
  %   \label{fig:H-edge}
  % \end{subfigure}
  \includegraphics[scale=1,page=6,height=.7\textwidth]{V-gadget}
  \caption{}
    \label{fig:HV-to-L}
  \end{subfigure}
  \caption{(a) Edge gadget for a V-edge. (b) Connections among gadgets.}
  \end{figure}

\newcommand{\proofoftheoremhardness}{
%\begin{proof}
  We reduce from HV-rectilinear planarity testing, which is \NP-hard
  even for biconnected graphs~\cite{dlp-chvrp-14}.  An instance of
  this problem is a degree-$4$ planar graph~$G$ where each edge is labeled
  either H or V.  The task is to decide whether~$G$ admits a planar
  orthogonal drawing (without bends) such that H-edges are drawn
  horizontally and V-edges are drawn vertically. We call such a drawing a planar \emph{HV-drawing}.
  
  Given a biconnected
  HV-graph~$G$, we construct an instance $G'$ of planar L-drawing by
  replacing each vertex by a $4$-wheel as in Fig.~\ref{fig:4wheel}, each edge labeled V (\emph{V-edge}) with the gadget shown in
  Fig.~\ref{fig:V-edge} and each edge labeled H (\emph{H-edge}) with the gadget shown in
  Fig.~\ref{fig:H-edge}. For a V-edge $(u,v)$, the two vertices of the edge gadget labeled $u$ and $v$ are identified with a V-port of the respective  vertex gadgets and for an H-edge with an H-port of the vertex gadgets.
  Obviously, this reduction is polynomial in the size of $G$. 
  
  Our high-level construction is somewhat similar to Br\"uckner's \NP-complete\-ness proof for \textsc{1-Embeddability} in the Kandinsky model~\cite[Theorem~3]{brueckner:ba13} in that we define gadgets that have a very limited flexibility in terms of their embeddings to realize horizontal and vertical edges. Yet the internals of the gadgets themselves and the reduction are quite different.

  We claim that $G'$ has a planar L-drawing if and only if $G$ has a planar HV-drawing. So first assume that $G'$ admits a planar L-drawing $\Gamma'$. We transform $\Gamma'$ into a planar HV-drawing. In a first step, we draw each vertex $v$ of $G$ at the position of the central vertex of the vertex gadget for $v$. Due to Lemma~\ref{cl:rect}, the edge gadgets attach to the bounding boxes of the vertex gadgets. Hence,
  for each edge $(u,v)$ of $G$, we can draw an orthogonal path from $u$ to $v$ by tracing the thick edges (red for a V-edge, blue for an H-edge) in its edge gadget and the two incident vertex gadgets  (see Fig.~\ref{fig:4wheel} and~\ref{fig:H-edge}).
  % Because this is a subdrawing of $\Gamma'$ it remains planar.
  This intermediate drawing as a subdrawing of $\Gamma'$ is a planar orthogonal drawing of $G$, where each edge is an 8-bend orthogonal staircase path with total rotation of 0. Using Tamassia's network flow model for orthogonal graph drawings~\cite{t-eggmn-87}, we can argue that an edge with rotation 0 is equivalent to a rectilinear edge without bends. In fact, the flow corresponding to the eight bends is cyclic and can be reduced to a flow of value 0, which implies no bends. We refer to Br\"uckner~\cite[Lemma~7]{brueckner:ba13} for more details of this argument.

  Now, conversely, assume that $G$ admits a planar HV-drawing $\Gamma$. In order to show that $\Gamma$ can be transformed into an L-drawing of $G'$ we first ``thicken'' $\Gamma$ by inflating vertices at grid points to squares and edges to corresponding rectangles, see Fig.~\ref{fig:HV-to-L}. This can easily be done without introducing any crossings of overlapping features by refining the grid on which $\Gamma$ is drawn. Since each vertex gadget in $G'$ can be drawn in a square (Fig.~\ref{fig:4wheel}) and each edge gadget in a rectangle (Figs.~\ref{fig:V-edge} and~\ref{fig:H-edge}), we can insert their drawings into the thickened drawing of $G$ as illustrated in Fig.~\ref{fig:HV-to-L}. This produces an L-drawing of $G'$.
  
  To see that the problem is in \NP, we note that for an embedding of a graph and a given orthogonal representation (see Tamassia~\cite{t-eggmn-87}) of that embedding, one can check whether all edges are represented as valid L-shapes in polynomial time.
  %\end{proof}
  }

\subsection{Fixed Embedding and Port Assignment}\label{SEC:ports}

In this section, we show how to decide efficiently whether there is a planar L-drawing for a plane directed graph with a fixed assignment of the edges to the four ports of the vertices.
Using Lemma~\ref{LEMMA:kandinsky} and the ILP formulation of Barth et
al.~\cite{barth_etal:gd06}, we first set up linear inequalities that
describe whether a plane $4$-modal graph has a planar L-drawing.  Using
these inequalities, we then transform our decision problem into a
matching problem that can be solved in linear time.

We call a vertex~$v$ an \emph{in/out-vertex} on a face~$f$ if $v$ is
incident to both, an incoming edge and an outgoing edge on $f$.
Let $x_{vf} \in \{0,1,2\}$ describe the angle in a face $f$ at a
vertex $v$: the angle between two outgoing or two incoming
edges is $x_{vf}\cdot \pi$ and the angle between an incoming and an outgoing
edge is $x_{vf}\cdot \pi + \pi/2$. Let $x_{f e}^v \in \{0,1\}$ be~1 if
there is a convex bend in face~$f$ on edge~$e$ assigned to a
vertex~$v$ to fulfill the bend-or-end property.  
%\todo{AW: Clarify what it means that ``the bend on edge~$e$ is
%  assigned to a vertex~$v$.''}
%
There is a planar L-drawing with these parameters if and
only if the following four conditions are satisfied (see~\springerarxiv{\cite{cccdnptw-plddg-tr-17}}{\ref{apx:ILP}} for details):
\begin{inparaenum}[(1)]
\item The angles around a vertex $v$ sum to $2\pi$.
\item\label{ITEM:oneBend} Each edge has exactly one bend.
\item\label{ITEM:faceBend} The number of convex angles minus the
  number of concave angles is $4$ in each inner face and $-4$ in the
  outer face.
\item The bend-or-end property is fulfilled, i.e., for any two
  edges~$e_1$ and~$e_2$ that are consecutive around a vertex~$v$ and
  that are both incoming or both outgoing, and for the faces~$f_1$,
  $f$, and~$f_2$ that are separated by~$e_1$ and~$e_2$ (in the cyclic order
  around~$v$), it holds that $x_{vf} + x_{f_1 e_1}^v + x_{f_2e_2}^v \geq 1$.
% \item The bend-or-end property is fulfilled, i.e., for any two
%   edges~$e_1$ and~$e_2$ that are consecutive around a vertex~$v$ and
%   that are both incoming or both outgoing, let~$f_1$, $f$, and~$f_2$
%   be the faces separated by~$e_1$ and~$e_2$ (in the cyclic order
%   around~$v$).  Then $x_{vf} + x_{f_1 e_1}^v + x_{f_2e_2}^v \geq 1$.
\end{inparaenum}
\noindent
Let $e=(v,w)$ be incident to faces $f$ and $h$, Condition (2) implies $-x_{h e}^v - x_{he}^w   = x_{fe}^v  + x_{fe}^w-1$. 
\mbox{Hence, (3) yields}

\smallskip
$(3') \displaystyle
  \sum_{\hidewidth\substack{e=(v,w) \textup{ incident~to~}
      f}\hidewidth} (x_{fe}^v + x_{fe}^w) - \sum_{v \textup{ on } f}
  x_{vf} = \pm 2 + (\# \textup{ in/out-vertices on }f - \deg f)/2$.
\smallskip

Observe that the number of in/out-vertices on a face $f$ is odd if and
only if $\deg f$ is odd. Moreover, if we omit the bend-or-end
property, we can formulate the remaining conditions as an
uncapacitated network flow problem. \label{L-flow} The network has
three types of nodes: one for each vertex, face, and edge of the
graph.  It has two types of edges: from vertices to incident faces and
from faces to incident edges.  The supplies are
$\lceil\frac{4-k}{2}\rceil$ for the $k$-modal vertices,
$\pm 2 + 1/2\cdot ( \# \textup{in/out-vertices } - \deg f)$ for a face
$f$, and $-1$ for the edges.

\begin{theorem}\label{THEO:assigned_ports}
  Given a directed plane graph $G$ and labels
  $\text{out}(e)\in\{\textup{top},\textup{bottom}\}$ and
  $\text{in}(e)\in\{\textup{right},\textup{left}\}$ for each edge~$e$, 
  it can be decided
  in linear time whether $G$ admits a planar L-drawing in which each
  edge $e$ leaves its tail at out$(e)$ and enters its head at in$(e)$.
\end{theorem}
\begin{sketch}  First, we
  have to check whether the cyclic order of the edges around a vertex
  is compatible with the labels. The labels determine the
  bends and the angles around the vertices, i.e., $x_{fe}^v+x_{fe}^w$
  for each edge $e=(v,w)$ and each incident face $f$, and $x_{vf}$ for
  each vertex $v$ and each incidence to a face $f$. We check whether
  these values fulfill Conditions~1, 2, and~$3'$. In order to also check
  Condition~4, we first assign for each port of a vertex $v$, all but
  the middle edges to $v$ (where a \emph{middle edge} of a port is the
  last edge in clockwise order bending to the left or the first edge
  bending to the right).  We check whether we thereby assign an edge
  more than once.  Assigning the middle edges can be reduced to a
  matching problem in a bipartite graph of maximum degree~2 where
  the nodes on one side are the ports with two middle edges and
  the nodes on the other side are the unassigned edges.
\end{sketch}

\section{Upward- and Upward-Rightward Planar L-Drawings}\label{se:fixed-embedding}

In this section, we characterize (see Theorem~\ref{th:characterization}) and construct (see Theorem~\ref{th:test-upward-upwardrightward}) upward-planar and upward-rightward planar L-drawings.

\subsection{A Characterization via Bitonic \st-Orderings}\label{se:char}

Characterizing the plane directed graphs that admit an L-drawing is an elusive goal.
%However, we identified two families of graphs admitting L-drawings with nice readability properties.
However, we can characterize two natural subclasses of planar L-drawings via bitonic \st-orderings. 

\begin{restatable}{theorem}{theoremcharacterization}
\label{th:characterization}
A planar \st-graph admits an upward- (upward-rightward-) planar L-drawing if and only if it admits a  bitonic (monotonically decreasing) pair. 
\end{restatable}

\begin{sketch}
  
  %
  %\noindent
  ``$\Rightarrow$'': Let $G=(V,E)$ be an \st-graph with $n$ vertices. The $y$-coordinates of an upward-
  (upward-rightward-) planar L-drawing of $G$ yield a bitonic
  (monotonically decreasing) \st-ordering.
  %We may assume that $\pi(V)=\{1,\dots,n\}$.
  \smallskip

  \noindent
  ``$\Leftarrow$'': Given a bitonic (monotonically decreasing)
  \st-ordering $\pi$ of~$G=(V,E)$, we construct an upward-
  (upward-rightward-) planar L-drawing of $G$ using an idea of
  Gronemann~\cite{g-bsupg-GD16}. For each vertex $v$, we use $\pi(v)$ as its $y$-coordinate.  

  For the $x$-coordinates we use a linear extension of a partial order $\prec$. 
  Let $v_1,\dots,v_n$ be the vertices of $G$ in the
  ordering given by $\pi$.  Let $G_i$ be the subgraph of~$G$ induced
  by $V_i=\{v_1,\dots,v_i\}$. To construct $\prec$,
  we augment~$G_i$ to $\overline{G}_i$ in such a way that the outer
  face~$f_{\overline{G}_i}$ of~$\overline{G}_i$ is a simple cycle and all vertices on
  $f_{\overline{G}_i}$ are comparable:
  We start with a triangle on $v_1$ and two new vertices $v_{-1}$
  and $v_{-2}$, with $y$-coordinates $-1$ and $-2$, respectively, and set $v_{-2} \prec v_1 \prec v_{-1}$. For
  $i=2,\dots,n$, let $u_1,\dots,u_k$ be the predecessors of $v_i$ in
  ascending order with respect to~$\prec$. If $\pi$ is monotonically
  decreasing or if $k=1$, we add an edge $e$ with head $v_i$. The
  tail of $e$ is the right neighbor $r$ of $u_k$ or the left neighbor $\ell$ of
  $u_1$ on $f_{\overline{G}_i}$, respectively, if the maximum successor $s_{\max}$ of
  $u_1$ is to the left (or equal to) or the right of $v_i$,
  respectively; see Fig.~\ref{fig:bitonic-a}. Now let $u_1,\dots,u_k$
  be the predecessors of $v_i$ in the possibly augmented graph; see
  Fig.~\ref{fig:bitonic-b}. We add the condition
  $u_{k-1} \prec v_i \prec u_k$.
\end{sketch}

\begin{figure}[tb]
  \begin{subfigure}{.45\textwidth}
    \centering
    \includegraphics[page=1]{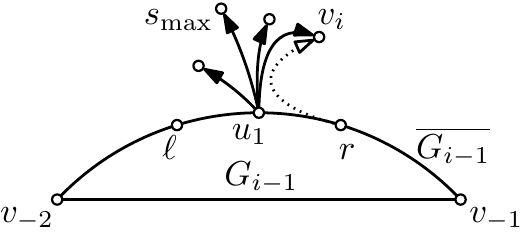}
    \vspace{-2ex}

    \caption{}
    \label{fig:bitonic-a}
  \end{subfigure}
  \hfill
  \begin{subfigure}{.45\textwidth}
    \centering
    \includegraphics{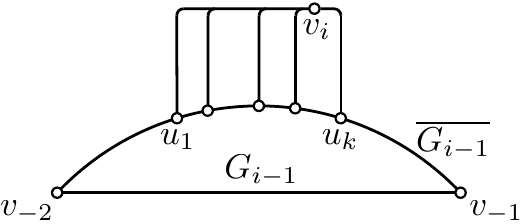}

    \vspace{-2ex}
    \caption{}
    \label{fig:bitonic-b}
  \end{subfigure}

  \caption{How to turn a bitonic \st-ordering into a planar L-drawing.}
  \label{fig:bitonic}
\end{figure}

\begin{corollary}\label{COR:undirected}
  Any undirected planar graph can be oriented such that it admits an
  upward-planar L-drawing.
\end{corollary}
\begin{proof}
  Triangulate the graph $G$ and construct a bitonic \st-ordering
  for undirected graphs~\cite{gronemann:gd14}. Orient the edges from
  smaller to larger \st-numbers.
\end{proof}

%\subsection{Upward Planar L-drawings and Bitonic \st-Orderings}
\subsection{Bitonic \st-Orderings in the Variable Embedding Setting}
\label{sse:directed}

By Theorem~\ref{th:characterization}, testing for the existence of an upward- (upward-rightward-) planar L-drawing of a planar \st-graph $G$ reduces to testing for the existence of a bitonic (monotonically decreasing) pair $\left< {\cal E}, \pi \right>$ for $G$. In this section, we give a linear-time algorithm to test an \st-graph for the existence of a bitonic pair $\left< {\cal E}, \pi \right>$.

\springerarxiv{The following lemma is proved in~\cite{cccdnptw-plddg-tr-17}.}{The following lemma is proved in~\ref{apx:R-node}.}

\begin{restatable}{lemma}{lemmanewsource}\label{le:edgeST}
  Let $G=(V,E)$ be a planar \st-graph with source~$s$, sink~$t$, and
  $(s,t) \notin E$.  Then there exists a supergraph $G'=(V',E')$
  of~$G$, where $V'=V \cup \{s'\}$ and $E' = E \cup \{(s',s), (s',t)\}$, such that (i) $G'$ is an \st-graph with source~$s'$ and sink~$t$, and (ii)~$G'$ admits a bitonic (resp., monotonically increasing) \st-ordering if and only if $G$ does.
  % Let $G=(V,E)$ be a planar \st-graph with source~$s$, sink~$t$, and
  % $(s,t) \notin E$.  Then there exists an supergraph $G'=(V',E')$
  % of~$G$ that is an st-graph with source~$s'$, sink~$t'$, and the two
  % properties: (i)~$(s',t') \in E'$ and (ii)~$G'$ admits a bitonic
  % (resp., monotonically increasing) \st-ordering if and only if $G$
  % does.
\end{restatable}

By Lemma~\ref{le:edgeST}, in the following we assume that an \st-graph
$G$ always contains edge $(s,t)$. Hence, either $G$ coincides with
edge $(s,t)$, which trivially admits a bitonic \st-ordering, or it is
biconnected.

A path $p$ from $u$ to $v$ in a directed graph is {\em monotonic increasing} ({\em monotonic decreasing}) if it is exclusively composed of forward (backward) edges. A path $p$ is {\em monotonic} if it is either monotonic increasing or monotonic decreasing.
A path $p$ with endpoints $u$ and $v$ is {\em bitonic} if it consists of 
a monotonic increasing path from $u$ to $w$ and of a monotonic decreasing path from $w$ to $v$; if $u \neq w$ and $v  \neq w$, then the path $p$ is {\em strictly bitonic} and $w$ is the {\em apex} of $p$.
An \st-graph $G$ is {\em $v$-monotonic}, {\em $v$-bitonic}, or  {\em strictly $v$-bitonic} if the subgraph of $G$ induced by the successors of $v$ is,
after the removal of possible transitive edges, a monotonic, bitonic, or strictly-bitonic path $p$, respectively. The apex of $p$, if any, is also called the {\em apex} of $v$ in $G$. 
If $p$ is monotonic and it is directed from $u$ to $w$, then vertices $u$ and $w$ are the {\em first successor of $v$ in $G$} and the  {\em last successor of $v$ in $G$}, respectively.
If $p$ is strictly bitonic, then its endpoints are the {\em first successors of $v$ in $G$}.
If $p$ consists of a single vertex, then such a vertex is both the {\em first and the last successor of $v$ in $G$}.
Let $G$ be an \st-graph and let $G^*$ be an \st-graph obtained by augmenting $G$ with directed edges. 
We say that the pair $\left<G, G^* \right>$ is {\em $v$-monotonic}, {\em $v$-bitonic}, or {\em strictly $v$-bitonic} if the subgraph of $G^*$ induced by the successors of $v$ in $G$ is, after the removal of possible transitive edges, a monotonic, bitonic, or strictly-bitonic path, respectively.

Although Gronemann~\cite{g-bsupg-GD16} didn't state this explicitly,
the following theorem immediately follows from the proof of his
Lemma~4.

\begin{theorem}[\hspace{-.01pt}\cite{g-bsupg-GD16}] \label{th:st-fixed-augmentation}
  A plane \st-graph $G=(V,E)$ admits a bitonic \st-ordering if and
  only if it can be augmented with directed edges to a planar \st-graph $G^*$ such that, for each
  vertex $v \in V$, the pair $\left<G, G^* \right>$ is $v$-bitonic. Further, any \st-ordering of $G^*$ is a bitonic \st-ordering of $G$.
\end{theorem}

In the remainder of the section, we show how to test in linear-time whether it is possible to augment a biconnected \st-graph $G$ to an \st-graph $G^*$ in such a way that the pair $\left<G, G^* \right>$ is $v$-bitonic, for any vertex $v$ of $G$. By virtue of Theorem~\ref{th:st-fixed-augmentation}, this allows us to test the existence of a bitonic pair $\left< {\cal E},\pi \right>$ for $G$. We perform a
bottom-up visit of the SPQR-tree $T$ of $G$ rooted at the reference
edge $(s,t)$ and show how to compute an augmentation for the pertinent
graph of each node $\mu \in T$ together with an embedding of it, if any exists.

Note that each vertex in an \st-graph is on a directed path from $s$
to $t$. Further, by the choice of the reference edge, neither $s$ nor
$t$ are internal vertices of the pertinent graph of any node of $T$. 
This leads to the next observation.

\begin{observation}
  \label{obs:child-orientation}
  For each node $\mu \in T$ with poles $u$ and $v$, the pertinent
  graph $\pert(\mu)$ of $\mu$ is an \st-graph whose source and sink
  are $u$ and $v$, or vice versa.
\end{observation}

Let $e$ be a virtual edge of $skel(\mu)$ corresponding to a node $\nu$ whose pertinent graph is an $st$-graph with source $s_\nu$ and sink $t_\nu$. By Observation~\ref{obs:child-orientation}, we say that $e$ {\em exits} $s_\nu$ and {\em enters} $t_\nu$.

The outline of the algorithm is as follows. Consider a node $\mu \in
T$ and suppose that, for each child $\mu_i$ of $\mu$, we have already
computed a pair $\left< \pert^*(\mu_i), {\cal E}_i^* \right>$ such
that $\pert^*(\mu_i)$ is an augmentation of $\pert(\mu_i)$, ${\cal
  E}_i^*$ is an embedding of $\pert^*(\mu_i)$, and $\left< \pert(\mu_i),\pert^*(\mu_i)\right>$ is $v$-bitonic, for each vertex $v$ of $\pert(\mu_i)$. 
  We show how to compute a pair $\left< \pert^*(\mu),
  {\cal E}^* \right>$ for node $\mu$, such that (i) the pair 
$\left< \pert(\mu),\pert^*(\mu)\right>$ is $v$-bitonic for each vertex $v$ in
$\pert(\mu)$, and (ii) the restriction of ${\cal E}^*$
to $\pert^*(\mu_i)$ is ${\cal E}_i^*$, up to a flip.  In the
following, for the sake of clarity, we first describe an overall quadratic-time
algorithm. We will refine this algorithm to run in linear time at the end of the section. 

For a node $\mu \in T$, we say that the pair $\left< \pert(\mu),\pert^*(\mu) \right>$ is of {\texttt Type B} if it is strictly $s_\mu$-bitonic and it is of {\texttt Type M} if it is $s_\mu$-monotonic. For simplicity, we also say that node
$\mu$ is of {\texttt Type B} or of {\texttt Type M} when, during the
traversal of $T$, we have constructed an augmentation $\pert^*(\mu)$
for $\mu$ such that $\left< \pert(\mu),\pert^*(\mu) \right>$ is of {\texttt Type B} or of {\texttt Type M}, respectively. 
%Figure~\ref{fig:replacement} shows an example where an augmentation
%    $\pert^*(\mu)$ for $\mu$ such that $\left<\pert(\mu),\pert^*(\mu)\right>$ is of {\texttt Type B} has been replaced with an augmentation
%    $\pert^+(\mu)$ for $\mu$ such that $\left<\pert(\mu),\pert^+(\mu)\right>$ is of~{\texttt Type~M} obtaining an augmentation that admits a bitonic \st-ordering if and only if the original augmentation does.
Figure~\ref{fig:replacement} shows an example where an augmentation
$G^*$ of $G$
    contains an augmentation
    $\pert^*(\mu)$ for $\mu$ which is replaced with an augmentation $\pert^+(\mu)$ such that
    $\left<\pert(\mu),\pert^*(\mu)\right>$ is of {\texttt Type~B},
    $\left<\pert(\mu),\pert^+(\mu)\right>$ is of~{\texttt Type~M},
    and $G^*$ admits a bitonic \st-ordering if and only if it still does after this replacement.
The following lemma formally shows that this type of replacement is always possible.

% shows that producing {\texttt Type M} nodes is preferable with respect to producing {\texttt Type B} nodes.

\begin{restatable}{lemma}{lemmapreferable}
  \label{le:preferable}
  Let $G$ be a biconnected \st-graph and let $G^*$ be an augmentation
  of $G$ such that $\left< G,G^* \right>$ is $v$-bitonic, for each vertex $v$ of $G$.  Consider
  a node $\mu$ of the SPQR-tree of $G$ and let $\pert^*(\mu)$ be the
  subgraph of $G^*$ induced by the vertices of $\pert(\mu)$. Suppose
  that $\left< \pert(\mu),\pert^*(\mu) \right>$ is of {\texttt Type B} and that $\pert(\mu)$
  also admits an augmentation $\pert^+(\mu)$
  such that $\left< \pert(\mu),\pert^+(\mu) \right>$ is of {\texttt Type M} and it is $v$-bitonic, for each vertex $v$ of $\pert(\mu)$.  There
  exists an augmentation $G^+$ of $G$ such that $\left< G,G^+ \right>$ is $v$-bitonic, for each vertex $v$ of
  $G$, and such that the subgraph of $G^+$ induced by
  the vertices of $\pert(\mu)$ is $\pert^+(\mu)$.
\end{restatable}

\begin{figure}[tb!]
  \centering
  \begin{subfigure}{.3\textwidth}
    \centering
    \includegraphics[page=1,width=.7\textwidth]{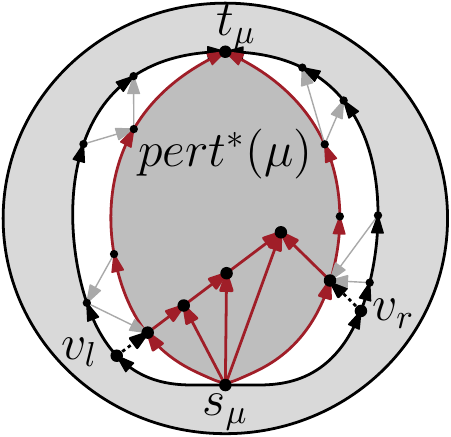}
    \caption{}
    \label{fig:replacement-a}
    %\label{fig:s-bitonic}
  \end{subfigure}
  \begin{subfigure}{.3\textwidth}
    \centering
    \includegraphics[page=2,width=.7\textwidth]{preference} 
    \caption{}
    \label{fig:replacement-b}
    %\label{fig:s-monotonic}
  \end{subfigure}
  \caption{Illustration for Lemma~\ref{le:preferable}.
  %Replacing an augmentation
    % $\pert^*(\mu)$ for $\mu$ such that $\left<\pert(\mu),\pert^*(\mu)\right>$ is of {\texttt Type B} with an augmentation
    % $\pert^+(\mu)$ for $\mu$ such that $\left<\pert(\mu),\pert^+(\mu)\right>$ is of~{\texttt Type~M}.
    }
  \label{fig:replacement}
\end{figure}

\newcommand{\proofoflemmapreferrable}{
First, observe that, by removing from $G^*$ all the edges (gray edges in Fig.~\ref{fig:replacement-a}) connecting a vertex in $\pert(\mu)$ that is not a successor of $s_\mu$ and a vertex not in $\pert(\mu)$ that is not a successor of $s_\mu$, we obtain an augmentation $G^\diamond$ of $G$ such that (i) the subgraph of $G^\diamond$ induced by the vertices of $\pert(\mu)$ is $\pert^*(\mu)$ and (ii) pair $\left< G,G^\diamond\right>$ is of $v$-bitonic, for any vertex $v$ of $G$\footnote{We remark that these edges are never introduced by our algorithm, however, for the sake of generality we make no assumption on their absence in this proof.}. Therefore, in the following we assume that $G^*=G^\diamond$.

Let $\cal E$ be a planar embedding of $G^*$.
Consider the subgraph $G^-_\mu$ obtained by removing from $G^*$ all the vertices of $V(\pert(\mu)) \setminus \{s_\mu,v_\mu\}$ an their incident edges. 
Let $\cal E^-_\mu$ be the planar embedding of $G^-_\mu$ induced by $\cal E$.
Let $f$ be the face of $\cal E^-_\mu$ whose boundary used to enclose the removed vertices. Observe that, the poles $s_\mu$ and $t_\mu$ of $\mu$ belong to $f$.
% Let $p$ be the strictly bitonic path induced in $\pert^*(\mu)$ by the successors of $s_\mu$ in $\pert(\mu)$.
Let $v_l$ and $v_r$ be successors of $s_\mu$ belonging to $G^-_\mu$ such that $v_l$ and $v_r$ are predecessors in $G^*$ of first successors of $\pert^*(\mu)$. Observe that, since we assumed $G^*=G^\diamond$, there exists exactly two vertices satisfying these properties.

Let $\cal E^+_\mu$ be a planar embedding of $\pert^+(\mu)$ in which $s_\mu$ and $t_\mu$ are incident to the outer face.
We now obtain plane graph $G^+=G^-_\mu \cup \pert^+(\mu)$ as follows. First, we embed $\pert^+(\mu)$ in the interior of $f$, identifying $s_\mu$ in $\pert^+(\mu)$ with $s_\mu$ in $f$ and $t_\mu$ in $\pert^+(\mu)$ with $t_\mu$ in $f$. Then, we insert two directed edges between a vertex in $G^-_\mu$ and a vertex of $\pert^*(\mu)$ as follows. We add a directed edge from $v_l$ to a first successor of $s_\mu$ in $\pert^+(\mu)$. Also, we add a directed edge from $v_r$ to the other first successor of $s_\mu$ in $\pert^+(\mu)$, if $\mu$ is not a Q-node, or to the same first successor of $s_\mu$ in $\pert^+(\mu)$ to which $v_l$ is now adjacent, otherwise.
%Let ${\cal E}^+$ be the obtained embedding of $G^+$.

To see that the directed graph $G^+$ is an \st-graph, observe that the added edges do not introduce any directed cycle as there
exists no directed path from a vertex in $\pert^+(\mu)$ to a vertex in $G^-_\mu$. Also, by construction, the subgraph of $G^+$ induced by
  the vertices of $\pert(\mu)$ is $\pert^+(\mu)$.

We now show that the pair $\left< G,G^+\right>$ is $v$-bitonic, for any $v$ in $G$.
Clearly, any vertex $v \notin \{s_\mu,v_l,v_r\}$ has the same successors in $G^+$ as in $G^*$, therefore $\left< G,G^+\right>$ is $v$-bitonic.
Further, by construction, $\left< G,G^+\right>$ is $s_\mu$-bitonic, that is, $\left< G,G^+\right>$ is of {\texttt Type B}; refer to Fig.~\ref{fig:replacement-b}.
Finally, since $v_l$ ($v_r$) is not adjacent in $G$ to any vertex in $\pert(\mu)$, the subgraph of $G^+$ induced by the successors of $v_l$ ($v_r$) in $G$ is the same as the subgraph of $G^*$ induced by the successors of $v_l$ ($v_r$) in $G$. This concludes the proof.
%\end{proof}
}

Consider a node $\mu$ of the SPQR-tree $T$ of $G$. We now show how to test the existence of a pair $\left< \pert^*(\mu), {\cal E}^* \right>$ such that (i) $\mu$ is
of {\texttt Type M} or, secondarily, of {\texttt Type B}, or report
that no such a pair exists, and (ii)  ${\cal E}^*$ is a planar embedding of $\pert^*(\mu)$.
In fact, by Lemma~\ref{le:preferable}, an embedding of $\mu$ of {\texttt Type M} would always be preferable to an embedding of {\texttt Type B}.

In any planar embedding ${\cal E}$ of $\pert(\mu)$ in which the poles are on the outer face $f_{out}$ of ${\cal E}$, we call {\em left path} ({\em right path}) of ${\cal E}$ the path that consists of the edges encountered in a clockwise traversal (in a counter-clockwise traversal) of the outer face of ${\cal E}$ from $s_\mu$ to $t_\mu$.

The following observation will prove useful to construct embedding ${\cal E}^*$. 

\begin{observation}\label{obs:ext-vert}
Let $\left< \pert^*(\mu), {\cal E}^* \right>$  be a pair such that $\left< \pert(\mu),\pert^*(\mu) \right>$ is $s_\mu$-bitonic and ${\cal E}^*$ is a planar embedding of $\pert^*(\mu)$ in which $s_\mu$ and $t_\mu$ lie on the external face. We have that:
\begin{enumerate}[(i)]
\item \label{obs:ext-vert-1} If $\mu$ is of {\texttt Type M}, then the first and the last successors of $s_\mu$ in $\pert^*(\mu)$ lie one on the left path and the other on the right path of ${\cal E}^*$. In particular, if the first and the last successor of $\mu$ are the same vertex, then such a vertex belongs to both the left path and the right path of ${\cal E}^*$.
\item \label{obs:ext-vert-2} If $\mu$ is of {\texttt Type B}, then the two first successors of $s_\mu$ in $\pert^*(\mu)$ lie one on the left path and the other on the right path of ${\cal E}^*$.
\end{enumerate}
\end{observation}

We distinguish four cases based on whether node $\mu$ is an S-, P-, Q-, or R-node.

{\smallskip\noindent{\bf Q-node.}} 
Here, $\left< \pert(\mu),\pert(\mu) \right>$  is trivially of {\texttt Type M}, i.e., $\pert^*(\mu) =
\pert(\mu)$.
%For a Q-node $\mu$, $\left< \pert(\mu),\pert(\mu) \right>$  is trivially of {\texttt Type M}. Hence, $\pert^*(\mu) =
%\pert(\mu)$.

{\smallskip\noindent{\bf S-node.}} Let $e_1,\dots,e_k$ be the virtual
edges of $\skel(\mu)$ in the order in which they appear from the
source $s_\mu$ to the target $t_\mu$ of $\skel(\mu)$, and let $\mu_1,\dots,\mu_k$ be
the corresponding children of $\mu$, respectively.  We obtain $\pert^*(\mu)$ by
replacing each virtual edge $e_i$ in $\skel(\mu)$ with $\pert^*(\mu_i)$. Also,
we obtain the embedding ${\cal E}^*$ by arbitrarily selecting a flip
for each embedding ${\cal E}_i^*$ of $\pert^*(\mu_i)$. Clearly, 
node $\mu$ is of {\texttt Type M} if and only if $\mu_1$ is of
{\texttt Type M} and it is of {\texttt Type B}, otherwise.

{\smallskip\noindent{\bf P-node.}} Let $e_1,\dots,e_k$ be the virtual
edges of $\skel(\mu)$ and let $\mu_1,\dots,\mu_k$ be the corresponding
children of $\mu$, respectively.

First, observe that if there exists more than one child of $\mu$ that
is of {\texttt Type B}, then node $\mu$ does not admit an
augmentation $\pert^*(\mu)$ where $\left< \pert(\mu),\pert^*(\mu) \right>$
is $s_\mu$-bitonic. In fact, if there exist two such nodes
$\mu_i$ and $\mu_j$, then both the subgraphs of $\pert^*(\mu_i)$ and
$\pert^*(\mu_j)$ induced by the successors of $s_\mu$ in $\pert(\mu_i)$ and
in $\pert(\mu_j)$, respectively, contain an apex vertex. This implies that $s_\mu$ would have more than one apex.

Second, observe that if there exists a child $\mu_i$ of $\mu$ of
{\texttt Type B} and the edge $(s_\mu,t_\mu)$ belongs to $\pert(\mu)$,
then node $\mu$ does not admit an augmentation $\pert^*(\mu)$ such that $\left< \pert(\mu),\pert^*(\mu) \right>$ is $s_\mu$-bitonic. In
fact, $\pert^*(\mu_i)$ contains a apex of $s_\mu$ different from
$t_\mu$; this is due to the fact that edge $(s_\mu,t_\mu) \notin \pert^*(\mu_i)$. Also, vertex
$t_\mu$ must be an apex of $s_\mu$ in any augmentation 
$\pert^*(\mu)$ of $\pert(\mu)$ such that $\left< \pert(\mu),\pert^*(\mu) \right>$ is $v$-bitonic, for each vertex $v$ of $\pert(\mu)$. Namely, any augmentation $\pert^*(\mu)$ of $\pert(\mu)$
yields an \st-graph with source $s_\mu$ and sink $t_\mu$ and,
as such, no directed path exits from $t_\mu$ in
$\pert^*(\mu)$. As for the observation in the previous paragraph, this implies that $s_\mu$ would have more than one apex.

We construct $\pert^*(\mu)$ as follows. We embed $\skel(\mu)$ in such a way that the edge $(s_\mu,t_\mu)$, if any, or the virtual edge corresponding to the unique child of $\mu$ that is of {\texttt Type B}, if any, is the right-most virtual edge in the embedding. Let $e_1,\dots,e_k$ be the virtual edges of $\skel(\mu)$ in the order in which they appear clockwise around $s_\mu$ in $\skel(\mu)$.
Then, for each child $\mu_i$ of $\mu$, we choose a flip of embedding ${\cal E}_i^*$ such that a first successor of $s_\mu$ in $\pert^*(\mu_i)$ lies along the left path of ${\cal E}_i^*$. Now, for $i=1,\dots,k-2$, we add an edge connecting the last successor of $s_\mu$ in $\pert^*(\mu_i)$
 and the first successor of $s_\mu$ in $\pert^*(\mu_{i+1})$. 
Finally, we possibly add an edge connecting the last successor $v_l$ of $s_\mu$ in $\pert^*(\mu_{k-1})$ and a suitable vertex in $\pert^*(\mu_{k})$. Namely, if a node $\mu_k$ is of {\texttt Type B}, then we add an edge between $v_l$ and the first successor of $s_\mu$ in $\pert^*(\mu_k)$ that lies along the left path of ${\cal E}_k^*$. If $\mu_k$ is of {\texttt Type M} and it is not a Q-node, then we add an edge between $v_l$ and the first successor of $s_\mu$ in $\pert^*(\mu_k)$. Otherwise $\pert^*(\mu_k) = (s_\mu,t_\mu)$ and we add the edge $(v_l,t_\mu)$ if no such an edge belongs to $\pert^*(\mu_{k-1})$.

 Observe that, the added edges do not introduce any directed cycle as there exists no directed path from a vertex in $\pert^*(\mu_{i+1})$ to a vertex in $\pert^*(\mu_i)$. Further, by Observation~\ref{obs:ext-vert} the added edges do not disrupt planarity. Therefore, the obtained augmentation $\pert^*(\mu)$ of $\pert(\mu)$ is, in fact, a planar \st-graph.

 Finally, we have that node $\mu$ is of {\texttt Type M} if and only if $\mu_k$ is of {\texttt Type M}.

{\smallskip\noindent{\bf R-node.}}
The case of an R-node $\mu$ is
% rather technical and lengthy and it is 
detailed in~\springerarxiv{\cite{cccdnptw-plddg-tr-17}.}{\ref{apx:R-node}.} 
For each node $v$ of $\skel(\mu)$, we have to consider the virtual edges $e_1,\dots,e_k$  of $\skel(\mu)$ exiting $v$ and the corresponding children $\mu_1,\dots,\mu_k$ of $\mu$, respectively. Similarly to the P-node case, we pursue an augmentation of $\pert(\mu)$ by inserting edges that connect $\pert(\mu_i)$ with $\pert(\mu_{i+1})$, with $i=1,\dots,k-1$.
Differently from the P-node case, however, more than one $\pert(\mu_i)$  may contain an edge between the poles of $\mu_i$. Further, also the faces of $\skel(\mu)$ may play a role, introducing additional constraints on the existence and the choice of the augmentation.

%This concludes the algorithm.

We have the following theorem.

\begin{theorem}\label{th:test-bitonic}
It is possible to decide in linear time whether a planar \st-graph $G$ admits a bitonic pair $\left< {\cal E}, \pi \right>$.
\end{theorem}

\begin{proof}
Let $\rho$ be the root of the SPQR-tree of $G$. The algorithm described above computes a pair $\left< \pert^*(\rho), {\cal E}^* \right>$ for $G$, if any exists, such that (i) the \st-graph $\pert^*(\rho)$ is an augmentation of $G$, (ii)  for any vertex $v$ of $G$, $\left< \pert(\rho),\pert^*(\rho) \right>$ is $v$-bitonic, and (iii) ${\cal E}^*$ is a planar embedding of $\pert^*(\rho)$. 
Let $\cal E$ be the restriction of ${\cal E}^*$ to $G$.
By Theorem~\ref{th:st-fixed-augmentation}, any \st-ordering $\pi$ of $\pert^*(\rho)$ is a bitonic \st-ordering of $G$ with respect to $\cal E$. Hence, $\left< {\cal E}, \pi \right>$ is a bitonic pair of~$G$.

We first show that the described algorithm has a quadratic running time. Then, we show how to refine it in order to run in linear time.
 For each node $\mu$ of $T$, the algorithm stores a pair $\left< \pert^*(\mu),  {\cal E} \right>$. Processing a node takes $O(|\pert^*(\mu)|)$ time. Since $|\pert^*(\mu)| \in O(|\pert(\mu)|)$, the overall running time is~$O(|G|^2)$. 

To achieve a linear running time,
observe that we do not need to compute the embeddings of the augmented pertinent graphs $\pert^*(\mu)$, for each node $\mu$ of $T$, during the bottom-up traversal of $T$. In fact, any embedding ${\cal E}^*$ of $\pert^*(\rho)$ yields an embedding ${\cal E}$ of $G$ such that $\pi$ is bitonic with respect to $\cal E$.  
To determine the endpoints of the augmenting edges, we only need to associate a constant amount of information with the nodes of $T$. Namely, for each node $\mu$ in $T$, we maintain (i) whether $\mu$ is of {\texttt Type B} or of {\texttt Type M}, (ii) if $\mu$ is of {\texttt Type M}, the first successor and the last successor of $s_\mu$ in $\pert^*(\mu)$, and (iii) if $\mu$ is of {\texttt Type B}, the two first successors of $s_\mu$ in $\pert^*(\mu)$. Therefore, processing a node takes $O(|\skel(\mu)|)$ time. 
Since the sum of the sizes of the skeletons of the nodes in $T$ is linear in the size of $G$~\cite{dt-opt-96}, the overall running time is linear. 
\end{proof}

%TODO corollario monotonic

\begin{corollary}\label{co:test-monotonic}
It is possible to decide in linear time whether a planar \st-graph $G$ admits a monotonically decreasing pair $\left< {\cal E}, \pi \right>$.
\end{corollary}

\begin{proof}
The statement immediately follows from the fact that, in the algorithm described in this section, when computing a pair $\left< \pert^*(\mu), {\cal E}^* \right>$ for each node $\mu$ in $T$, a pair $\left< \pert(\mu),\pert^*(\mu) \right>$ of {\texttt Type M} is built whenever possible. Therefore, rejecting instances for which a pair $\left< \pert(\mu),\pert^*(\mu) \right>$ of {\texttt Type B} is needed yields the desired algorithm. 
\end{proof}

In conclusion, we have the following main result.

\begin{theorem}\label{th:test-upward-upwardrightward}
  It can be tested in linear time whether a planar \st-graph admits
  an upward- (upward-rightward-) planar L-drawing, and if so, such a drawing can be constructed in linear time. 
\end{theorem}
\begin{proof}
  We first test in linear time whether a planar \st-graph admits a bitonic pair (Theorem~\ref{th:test-bitonic}) or a monotonically decreasing pair (Corollary~\ref{co:test-monotonic}).
  Then, Theorem~\ref{th:characterization} shows how to construct in linear time an upward- (upward-rightward-) planar L-drawing from a bitonic (monotonically decreasing) pair.\end{proof}

\section{Open Problems}\label{se:open}

Several interesting questions are left open:
%\begin{inparaenum}
%\item 
Can we efficiently test  whether a directed plane
  graph admits a planar L-drawing?
%\item 
Can we efficiently recognize the directed graphs that are edge
  maximal subject to having a planar L-drawing 
  %(see~\ref{se:proofsopen} for density bounds)?
  %(A graph with $n$ vertices that admits a planar, upward-planar, or upward-rightward-planar L-drawing has at most $4n-6$, $3n-6$, or $2n-3$ edges and these bounds are tight.) ?
  (they have at most $4n-6$ edges where $n$ is the number of vertices---see~\springerarxiv{\cite{cccdnptw-plddg-tr-17}?}{\ref{se:proofsopen})?}
%\item 
Does every upward-planar graph have a (not necessarily upward-)
  planar L-drawing? %(In this case, the flow network described on
  %page~\pageref{L-flow} always admits a flow. Is it also possible to
  %add the bend-or-end property?)
%\item 
Can we extend the algorithm for computing a bitonic pair in the variable embedding setting to single-source multi-sink di-graphs?
%\item 
Does every bimodal graph have a planar L-drawing?
%\end{inparaenum}

\bibliographystyle{splncs03}
\bibliography{abbrv,eldrawings}

% TITTO: THIS IS TO HAVE A LARGER PAGE IN THE APPENDIX
\ifarxive
\clearpage
\appendix
%\newgeometry{left=2cm,right=2cm,top=2cm,bottom=2cm}
\else
\end{document}
\fi

\renewcommand{\thesection}{Appendix~\Alph{section}}

\section{SPQR Trees}\label{apx:spqr}

In this appendix we describe \emph{SPQR-trees}, a data structure introduced by Di Battista and Tamassia (see, e.g.,~\cite{dt-ogasp-90}) which allows to handle
%the decomposition of an \st-biconnectible \todo[inline]{verificare} graph into its triconnected components. 
the planar embeddings of an \st-biconnectible planar graph.

A graph is \emph{\st-biconnectible} if adding the edge $(s,t)$ yields a biconnected graph. Let $G$ be an \st-biconnectible graph. A \emph{separation pair} of $G$ is a pair of vertices whose removal disconnects the graph. A \emph{split pair} of $G$ is either a separation pair or a pair of adjacent vertices. A \emph{maximal split component} of $G$ with respect to a split pair $\{u, v\}$ (or, simply, a maximal split component of $\{u, v\}$) is either an edge $(u, v)$ or a maximal subgraph $G'$ of $G$ such that $G'$ contains $u$ and $v$, and $\{u, v\}$ is not a split pair of $G'$. A vertex $w \neq u,v$ belongs to exactly one maximal split component of $\{u, v\}$. We call \emph{split component} of $\{u, v\}$ the union of any number of maximal split components of $\{u, v\}$.

In this paper, we will assume that any SPQR-tree of a graph $G$ is rooted at one
edge of $G$, called \emph{reference edge}.

The rooted SPQR-tree $\mathcal{T}$ of a biconnected graph $G$, with respect to a reference edge $e$, describes a recursive decomposition of $G$ induced by its split pairs. The nodes of $\mathcal{T}$ are of four types: S, P, Q, and~R. Their connections are called \emph{arcs}, in order to distinguish them from the edges of $G$.

Each node $\mu$ of $\mathcal{T}$ has an associated \st-biconnectible multigraph, called the \emph{skeleton} of $\mu$ and denoted by $\skel(\mu)$. Skeleton $\skel(\mu)$ shows how the children of $\mu$, represented by ``virtual edges'', are arranged into $\mu$. The virtual edge in $\skel(\mu)$ associated with a child node $\nu$, is called the \emph{virtual edge of $\nu$ in $\skel(\mu)$}.

For each virtual edge $e_i$ of $\skel(\mu)$, recursively replace $e_i$ with the skeleton $\skel(\mu_i)$ of its corresponding child $\mu_i$. The subgraph of $G$ that is obtained in this way is the \emph{pertinent graph} of $\mu$ and is denoted by $G(\mu)$.

Given a biconnected graph $G$ and a reference edge $e=(u',v')$, the SPQR-tree $\mathcal{T}$ is recursively defined as follows. At each step, a split component $G^*$, a pair of vertices $\{u,v\}$, and a node $\nu$ in $\mathcal{T}$ are given. A node $\mu$ corresponding to $G^*$ is introduced in $\mathcal{T}$ and attached to its parent $\nu$. Vertices $u$ and $v$ are the \emph{poles} of $\mu$ and denoted by $u(\mu)$ and $v(\mu)$, respectively. The decomposition possibly recurs on some split components of $G^*$. At the beginning of the decomposition $G^* = G - \{e\}$, $\{u,v\}=\{u',v'\}$, and $\nu$ is a Q-node corresponding to $e$.

\begin{description}

\item[\em{Base Case:}] If $G^*$ consists of exactly one edge between $u$ and $v$, then $\mu$ is a Q-node whose skeleton is $G^*$ itself.

\item[\em{Parallel Case:}] If $G^*$ is composed of at least two maximal split components $G_1, \dots, G_{k}$ ($k \geq 2$) of $G$ with respect to $\{u,v\}$, then $\mu$ is a P-node. The graph $\skel(\mu)$ consists of $k$ parallel virtual edges between $u$ and $v$, denoted by $e_1, \dots, e_{k}$ and corresponding to $G_1, \dots, G_{k}$, respectively. The decomposition recurs on $G_1, \dots, G_{k}$, with $\{u,v\}$ as pair of vertices for every graph, and with $\mu$ as parent node.

\item[\em{Series Case:}] If $G^*$ is composed of exactly one maximal split component of $G$ with respect to $\{u,v\}$ and if $G^*$ has cut vertices $c_1, \dots, c_{k-1}$ ($k \geq 2$), appearing in this order on a path from $u$ to $v$, then $\mu$ is an S-node. Graph $\skel(\mu)$ is the path $e_1, \dots, e_k$, where virtual edge $e_i$ connects $c_{i-1}$ with $c_i$ ($i = 2, \dots ,k-1$), $e_1$ connects $u$ with $c_1$, and $e_k$ connects $c_{k-1}$ with $v$. The decomposition recurs on the split components corresponding to each of $e_1, e_2,\dots, e_{k-1}, e_{k}$ with $\mu$ as parent node, and with $\{u,c_1\}, \{c_1,c_2\},$ $\dots,$ $\{c_{k-2},c_{k-1}\}, \{c_{k-1},v\}$ as pair of vertices, respectively.

\item[\em{Rigid Case:}] If none of the above cases applies, the purpose of the decomposition step is that of partitioning $G^*$ into the minimum number of split components and recurring on each of them. We need some further definition. Given a maximal split component $G'$ of a split pair $\{s,t\}$ of $G^*$, a vertex $w \in G'$ \emph{properly belongs} to $G'$ if $w \neq s, t$. Given a split pair $\{s,t\}$ of $G^*$, a maximal split component $G'$ of $\{s,t\}$ is \emph{internal} if neither $u$ nor $v$ (the poles of $G^*$) properly belongs to~$G'$, \emph{external} otherwise. A \emph{maximal split pair} $\{s,t\}$ of $G^*$ is a split pair of $G^*$ that is not contained in an internal maximal split component of any other split pair $\{s',t'\}$ of $G^*$. Let $\{u_1,v_1\}, \dots, \{u_k,v_k\}$ be the maximal split pairs of $G^*$ ($k \geq 1$) and, for $i = 1, \dots, k$, let $G_i$ be the union of all the internal maximal split components of $\{u_i,v_i\}$. Observe that each vertex of $G^*$ either properly belongs to
exactly one $G_i$ or belongs to some maximal split pair $\{u_i,v_i\}$. The node $\mu$ is an R-node. The graph $\skel(\mu)$ is the graph obtained from $G^*$ by replacing each subgraph $G_i$ with the virtual edge $e_i$ between $u_i$ and $v_i$. The decomposition recurs on each $G_i$ with $\mu$ as parent node and with $\{u_i,v_i\}$ as pair of vertices.
\end{description}

For each node $\mu$ of $\mathcal{T}$ with poles $u$ and $v$, the construction of $\skel(\mu)$ is completed by adding a virtual edge $(u,v)$ representing the \emph{rest of the graph}, that is, the graph obtained from $G$ by removing all the vertices of $G(\mu)$, except for its poles, together with their incident edges.

The SPQR-tree $\mathcal{T}$ of a graph $G$ with $n$ vertices and $m$ edges
has $m$ Q-nodes and $O(n)$ S-, P-, and R-nodes. Also, the total number of
vertices of the skeletons stored at the nodes of $\mathcal{T}$ is $O(n)$.
Finally, SPQR-trees can be constructed and handled efficiently. Namely,
given a biconnected planar graph $G$, the SPQR-tree $\mathcal{T}$ of $G$
can be computed in linear time~\cite{dt-omtc-96,dt-opt-96,gm-lti-00}.

\section{Omitted Proofs}\label{apx:proofs}

In this appendix we give full versions of sketched or omitted proofs.

\subsection{Omitted Proofs of Section~\ref{sse:undirected}}

As a central building block for our hardness reduction we use a graph $W$ that can be constructed starting from a $4$-wheel with central vertex $c$ and rim $(u,v,w,z)$ such that vertices $v$ and $z$ are sinks and $u$ and $w$ are sources, by adding edges $(v,c)$, $(z,c)$, $(c,w)$, and $(c,u)$; see Fig.~\ref{fig:4wheel}. Note that the edges incident to $c$ come in pairs of both directions. We denote the vertices $v$ and $z$ as \emph{V-ports} of W and the vertices $u$ and $w$ as \emph{H-ports}. We first study
the properties of planar L-drawings of $W$.

\lemmawheel*

\begin{proof}
  In any orthogonal drawing of $W$, the outer cycle $(u,v,w,z)$ forms an orthogonal polygon $P$ with at least four convex corners. Since any two consecutive edges on the outer cycle have the same direction with respect to their common vertex $r \in \{u,v,w,z\}$, i.e., they are either both incoming or outgoing at $r$, they must use the same port or two opposite ports of $r$. In fact, if they would use the same port, they would form an angle of $2\pi$ in the outer face and force the edge $(r,c)$ to use the very same port. This, however, would imply that all three edges incident to $r$ have the same direction, which is a contradiction. Hence each of the four outer vertices has an angle of $\pi$ in the outer face and cannot form a convex corner of $P$. 
  
  Since there are four edges on the outer cycle, each of which has exactly one bend, this immediately implies that $P$ is a rectangle whose corners are formed by the bends of the four edges of the outer face and each of the four vertices of the outer face must lie on one of the rectangle sides. The remaining edges to $c$ use the port inside $P$, consistently bend once (left or right) from the perspective of $c$, and then connect to $c$ from all four sides. Figure~\ref{fig:4wheel} shows an example.
\end{proof}

\theoremhardness*

\begin{proof}
\proofoftheoremhardness
\end{proof}

  \begin{figure}[t]
    \centering
    \includegraphics[scale=1,page=2]{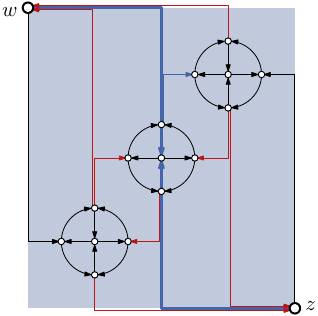}
    \caption{Edge gadget for an H-edge.}
    \label{fig:H-edge}
  \end{figure}

We remark that the graph $G'$ that we construct in our reduction is a simple directed graph. With the exception of the four spoke edges of the wheel graph~$W$ (see Fig.~\ref{fig:4wheel}) each underlying undirected would not have multi-edges. It is not difficult to extend our reduction so that the red and blue edges in Fig.~\ref{fig:4wheel} are removed from the gadget and the entire graph $G'$ becomes an oriented graph, i.e., a graph without 2-cycles. In that case, however, when we construct the intermediate staircase paths for the edges of the HV-drawing, we still use the removed ``mirrored'' L-shape for the first and last two segments of each edge path, which is always possible without crossings in any L-drawing of~$W$.  

\subsection{Omitted Proofs of Section~\ref{SEC:ports} Including the Relation with the Kandinsky Model}\label{apx:ILP}

\setcounterref{lemma}{LEMMA:kandinsky}
\addtocounter{lemma}{-1}
\begin{lemma}
  A graph has a planar L-drawing if and only if it admits a drawing in
  the Kandinsky model with the following properties
  \begin{enumerate}
  \item
    Each edge bends exactly once.
  \item
    At each vertex, the angle between two outgoing (or between two incoming) edges is 0 or $\pi$.
  \item
    At each vertex, the angle between an incoming edge and an outgoing edge is $\pi/2$ or $3\pi/2$.
  \end{enumerate}
\end{lemma}
\begin{proofwithwrapfig}
  Given a drawing in the Kandinsky model that meets Conditions~1-3, we
  can bundle the edges on the finer grid to lie on the coarser
  grid. It remains to perturb the coordinates such that the $x$- and
  $y$-coordinates, respectively, of the vertices are distinct: Assume,
  two vertices $v$ and $w$ have the same $y$-coordinate.  Let $\delta>0$
  be the minimum difference in $y$-coordinates between $v$ and any
  vertex or segment above $v$.  Since all edges have one bend, we can
  shift $v$ upward by $\delta/2$---changing only the drawing of edges
  incident to $v$. Doing this iteratively yields a planar
  L-drawing---or a rotation of $\pi/2$ of it.

  \begin{wrapfigure}[7]{R}{.30\textwidth}
    \centering
    \vspace{-1.0cm}
    \includegraphics[page=1]{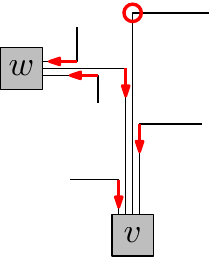}
    % \caption{Furthest Bend.}\label{FIG:kandinskyFurthestBend}
  \end{wrapfigure}
  Given a planar L-drawing, we can distribute the edges on the finer
  grid maintaining the embedding. Since all vertices have distinct $x$-
  and $y$-coordinates, there are no empty faces. It remains to assign
  the bends to the vertices in order 
  to fulfill the bend-or-end
  property: For each port 
  (top, right, bottom, left) 
  of a vertex $v$, we assign all bends of
  incident edges, but the furthest to $v$ (see
  the figure on the right---the furthest bend of top of $v$ is
  encircled).
  Observe that if the bend on an edge $\{v,w\}$ is not a furthest bend
  for $v$ then it is a furthest bend for $w$.
  Thus, no bend will be assigned to two vertices.
\end{proofwithwrapfig}

\subsubsection{ILP formulation for the proof of Theorem~\ref{THEO:assigned_ports}}

\begin{enumerate}
\item The angles around a vertex $v$ sum to $2\pi$:
$\displaystyle
\sum_{f \textup{ incident } v} x_{vf} = \left\{\begin{array}{cl}
 2  \textup{ if } v \textup{ 1-modal}\\
 1  \textup{ if } v \textup{ 2-modal}\\
 0  \textup{ if } v \textup{ 4-modal}
\end{array}\right. 
$
\item\label{ITEM:oneBend:A} All edges are bent exactly once, i.e., for
  each edge $e=\{v,w\}$ separating the faces~$f$ and~$h$, we have
$
       x_{f e}^v + x_{he}^v + x_{fe}^w + x_{he}^w = 1.
$ 
\item\label{ITEM:faceBend:A} The number of convex angles minus the
  number of concave angles is 4 in each inner face and $-4$ in the
  outer face, i.e., for each face~$f$, we have
\[
\sum_{\hidewidth\substack{e=\{v,w\}\\\textup{ separating  }\\ f \textup{ and } h}\hidewidth} 
(x_{fe}^v - x_{he}^v + x_{fe}^w - x_{he}^w) + 
\sum_{\hidewidth\substack{v \textup{ on } f\\ \textup{not in/out}}\hidewidth} (2-2x_{vf}) +  \sum_{\hidewidth\substack{v \textup{ on } f\\ \textup{in/out}}\hidewidth} (2-(2x_{vf}+1)) = \pm 4.
\]
\item The bend-or-end property is fulfilled, i.e., for any two
  edges~$e_1$ and~$e_2$ that are consecutive around a vertex~$v$ and
  that are both incoming or both outgoing, and for the faces~$f_1$,
  $f$, and~$f_2$ that are separated by~$e_1$ and~$e_2$ (in the cyclic order
  around~$v$), it holds that $x_{vf} + x_{f_1 e_1}^v + x_{f_2e_2}^v \geq 1$.
% \item The bend-or-end property is fulfilled, i.e., for any two
%   edges~$e_1$ and~$e_2$ that are consecutive around a vertex~$v$ and
%   that are both incoming or both outgoing, let~$f_1$, $f$, and~$f_2$
%   be the faces separated by~$e_1$ and~$e_2$ (in the cyclic order
%   around~$v$).  Then $x_{vf} + x_{f_1 e_1}^v + x_{f_2e_2}^v \geq 1$.
\end{enumerate}
\noindent
Observe that (\ref{ITEM:oneBend:A}) implies $-x_{h e}^v - x_{he}^w   = x_{fe}^v  + x_{fe}^w-1$. Hence, (\ref{ITEM:faceBend:A}) yields
\begin{enumerate}
\item[\ref{ITEM:faceBend:A}'.]  $\displaystyle
  \sum_{\hidewidth\substack{e=\{v,w\} \\\textup{ incident }
      f}\hidewidth} (x_{fe}^v + x_{fe}^w) - \sum_{v \textup{ on } f}
  x_{vf} = \pm 2 + (\# \textup{ in/out-vertices on }f - \deg f)/2$.
\end{enumerate}

\setcounterref{theorem}{THEO:assigned_ports}
\addtocounter{theorem}{-1}
\begin{theorem}
  Given a directed plane graph $G$ and labels
  $\text{out}(e)\in\{\textup{top},\textup{bottom}\}$ and 
  $\text{in}(e)\in\{\textup{right},\textup{left}\}$
  for each edge~$e$, 
  it
  can be decided in linear time whether $G$ admits a planar L-drawing
  in which each edge $e$ leaves its tail at out$(e)$ and enters its
  head at in$(e)$.
\end{theorem}
\begin{proof}
  Observe that the labeling determines the bends, i.e., the value
  $x_{fe}^v+x_{fe}^w$ for each edge $e=(v,w)$ and each incident face
  $f$. First, we have to check whether the cyclic order of the edges
  around a vertex is compatible with the labels, i.e., in clockwise
  order we have outgoing edges labeled (top,$\cdot $), incoming edges
  labeled ($\cdot $,left), outgoing edges labeled (bottom,$\cdot $),
  and incoming edges labeled ($\cdot $,right). For a fixed port, edges
  bending to the right must precede edges bending to the left. We call
  an edge a \emph{middle edge} of a port if it is the last edge
  bending to the left or the first edge bending to the right. Observe
  that each port has zero, one, or two middle edges.

  If the compatibility check does not fail then the labels also
  determine the angles around the vertices, i.e., the variables
  $x_{vf}$ for each vertex $v$ and each incidence to a face $f$. Now, we
  check whether these values fulfill Conditions 1, 2, and 3'. %, i.e.,
  %whether they yield a flow.

  Finally, we have to check, whether Condition~4, i.e., the
  bend-or-end property can be fulfilled. To this end, we have to
  assign edges with concave bends to zero angles at an incident vertex
  in the same face. We must assign for each port of a vertex $v$, all
  but the middle edges to $v$. If at this stage an edge is assigned to
  two vertices, then $G$ does not admit a planar L-drawing with the
  given port assignment. Otherwise, it remains to deal with the zero
  angles between two middle edges of a port. To this end, consider the
  following graph $B$. The nodes are on one hand the ports with two
  middle edges and on the other hand the edges that are middle edges
  of at least one port and that are not yet assigned to a vertex. A
  port of a vertex $v$ and an edge $e$ are adjacent in $B$ if and only
  if $e$ is a middle edge of $v$. Observe that $B$ is a bipartite
  graph of maximum degree two and, thus, consists of paths, even
  length cycles, and isolated vertices. We have to test whether $B$
  has a matching in which every port node is matched. This is true if and
  only if no port is isolated and there is no maximal path starting
  and ending at a port node.
\end{proof}

\subsection{Omitted proofs of Sect.~\ref{se:char}}

\theoremcharacterization*

\begin{proof}
  Let $G=(V,E)$ be a planar \st-graph with $n$ vertices.

  \noindent
  ``$\Rightarrow$'': The $y$-coordinates of an upward-
  (upward-rightward-) planar L-drawing of~$G$ yield a bitonic
  (monotonically decreasing) \st-ordering $\pi$ with respect to the
  embedding $\mathcal E$ given by the L-drawing.  
  %We may assume that $\pi(V)=\{1,\dots,n\}$.
  \smallskip

  \noindent
  ``$\Leftarrow$'': Given a bitonic (monotonically decreasing) \st-ordering $\pi$ of~$G$, we
  construct an upward- (upward-rightward-) planar L-drawing of $G$ using an idea of
  Gronemann~\cite{g-bsupg-GD16}.  For $i=1,\dots,n$, let $v_i \in V$
  be the vertex with $\pi(v_i)=i$, set the $y$-coordinate of~$v_i$
  to~$i$, and let~$G_i$ be the subgraph of~$G$
  induced by $V_i=\{v_1,\dots,v_i\}$.

  For the $x$-coordinates we construct a partial order $\prec$ in such a
  way that, for $i=2,\dots,n$, all vertices on the outer face of $G_i$
  are comparable and the L-drawing of $G_i$ is planar, embedding
  preserving, and has the property that any edge from~$V_i$ to~$V
  \setminus V_i$ can be added upward and in an embedding preserving
  way, no matter how we choose the $x$-coordinates of
  $v_{i+1},\dots,v_n$.

  During the construction, we augment~$G_i$ to $\overline{G}_i$ in such a way that the
  outer face~$f_{\overline{G}_i}$ of~$\overline{G}_i$ is a simple cycle.  We start by
  adding two artificial vertices $v_{-1}$ and $v_{-2}$ with
  $y$-coordinates $-1$ and $-2$, respectively, that are connected
  to~$v_1$ and to each other.  We set $v_{-2} \prec v_1 \prec v_{-1}$.
  Now let $i \in \{2,\dots,n\}$ % Used to be $2,\dots,n-1$ -- why?
  and assume that we have already fixed the relative coordinates
  of~$G_{i-1}$.  Let $u_1,\dots,u_k$ be the predecessors of $v_i$ in
  ascending order with respect to~$\prec$.

  If $\pi$ is monotonically decreasing or if $k=1$, we first augment
  the graph.  In the former case, we add to~$G$ an edge between~$v_i$
  and the right neighbor of~$u_k$ on~$f_{\overline{G}_{i-1}}$.
  In the latter case, let~$\ell$ and~$r$ be the left and the
  right neighbor of~$u_1$ on~$f_{\overline{G}_{i-1}}$, respectively; see
  Fig.~\ref{fig:bitonic-a}.  Following Gronemann~\cite{g-bsupg-GD16},
  we add a dummy edge from either~$\ell$ or~$r$ to~$v_i$:
  Let~$s_{\max}$ be
  the successor of~$u_1$ of maximum rank. We go in the circular order
  of the edges around~$u_1$ from~$u_1v_i$ to the left.  If we
  hit~$u_1s_{\max}$ before~$u_1\ell$, we insert the edge~$rv_i$
  into~$G$, otherwise the edge~$\ell v_i$.  Note that
  inserting the dummy edge does not violate planarity since, on that
  side, $u_k$ does not have any outgoing edge between $u_kv_i$
  and~$f_{\overline{G}_{i-1}}$.

  We now extend $\prec$.  Let $u_1,\dots,u_k$ be the $k \geq 2$
  predecessors of $v_i$ in the possibly augmented graph; see
  Fig.~\ref{fig:bitonic-b}.  Since $G$ has a sink only on the outer
  face, we can place $v_i$ anywhere between $u_1$ and $u_k$.  Adding
  the two conditions $u_{k-1} \prec v_i \prec u_k$ also sure that all
  edges except $(u_k,v_i)$ are rightward. But $(u_k,v_i)$ was
  introduced only as a dummy edge for the case of a monotonically
  decreasing~$\pi$.

  Any linear order that is compatible with $\prec$ yields unique
  $x$-coordinates in $\{1,\dots,n\}$ for the vertices of~$G$.  Together
  with the $y$-coordinates that we fixed above, we now have positions
  for the vertices in an upward- (upward-rightward-) planar L-drawing
  of~$G$.  Finally, we remove the dummy edges that we inserted earlier.
\end{proof}

\subsection{Omitted Proofs of Section~\ref{sse:directed}}\label{apx:R-node}

\lemmanewsource*

\begin{proof}
  % Graph $G'$ is defined as follows.  We set $V'=V \cup \{s'\}$ and $E'
  % = E \cup \{(s',s), (s',t)\}$.  It is easy to see that $G'$ is an
  % \st-graph whose source and sink are $s'$ and $t$, respectively.

  We prove the {\em if} direction. Let $\pi'$ be a bitonic (resp.,
  monotonically increasing) \st-ordering of $G'$ and let $\cal E'$ be
  a planar embedding of $G'$ compatible with $\pi$. We construct a
  ranking $\pi: V \rightarrow \{1,\dots,|V|\}$ by setting
  $\pi(v)=\pi'(v)-1$, for each $v \in V$. Also, we set ${\cal E}$ to
  the restriction of $\mathcal E'$ to $G$. Clearly, $\pi$ is a bitonic
  (resp., monotonically increasing) \st-ordering of~$G$ that is
  consistent with $\cal E$.

  We now prove the {\em only if}
  direction. Let $\pi$ be a bitonic (resp., monotonically increasing)
  \st-ordering of $G$ and let $\cal E$ be a planar embedding of $G$
  compatible with $\pi$.  We construct a ranking $\pi'= V' \rightarrow
  \{1,\dots,|V'|\}$ as follows: We set (i) $\pi'(s')=1$ and (ii)
  $\pi'(v)=\pi(v)+1$, for each $v \in V$.  We construct a planar
  embedding $\cal E'$ of $G'$ starting from $\cal E$ by drawing $s'$
  in the outer face of $\cal E$ and by routing edge $(s',t)$ so that
  vertex $t$ is the right-most successor of $s'$ in the left-to-right
  order of the successors of $s'$ around $s'$.
  We show that $\pi'$ is a bitonic (resp., monotonically increasing)
  \st-ordering of $G'$ and that $\cal E'$ is consistent with
  $\pi'$. Since, for each vertex $v \in V$, the ranks of the
  successors of $v$ in $\pi'$ have all been decreased by $1$ and since
  the left-to-right order of the successors of $v$ is the same in
  $\cal E'$ as in $\cal E$, it follows that such ranks form a bitonic
  (resp., monotonically increasing) sequence in $\pi'$ if and only if
  they do so in $\pi$. Also, $s$ and $t$ are the successors of $s'$ and
  $\pi(s)<\pi(t)$. Hence, the ranks of the successors of $s'$ form a
  monotonically increasing sequence. This concludes the proof of the
  lemma.
\end{proof}

\lemmapreferable*

\begin{proof}
\proofoflemmapreferrable
\end{proof}

\subsubsection{Details for the R-node Case}$ $

{\smallskip\noindent{\bf R-node.}} 
Recall that, by Observation~\ref{obs:child-orientation}, the skeleton of a node $\mu$ of $T$ is an \st-graph between its poles $s_\mu$ and~$t_\mu$. 

For each vertex $v \neq t_\mu$, let $e_1,\dots,e_k$ be the virtual edges exiting $v$ in the order in which they appear clockwise around $v$ in $\skel(\mu)$, and let $\mu_i$ be the node of $T$ corresponding to $e_i$.
First, observe that if there exists more than one virtual edge $e_i$ exiting from $v$ whose corresponding child $\mu_i$  is of {\texttt Type B}, then node $\mu$ does not admit an augmentation $\pert^*(\mu)$ such that $\left< \pert(\mu),\pert^*(\mu) \right>$ is $s_\mu$-bitonic. In fact, as shown for the P-node case, this implies that $s_\mu$ would have more than one apex.
We aim at (i) selecting a flip for each $\pert^*(\mu_i)$ and (ii) adding an edge between a vertex in $\pert^*(\mu_i)$ and a vertex in $\pert^*(\mu_{i+1})$, with $i=1,\dots,k-1$, in order to obtain an augmentation $\pert^*(\mu)$ of $\pert(\mu)$  such that $\left< \pert(\mu),\pert^*(\mu) \right>$ is $s_\mu$-bitonic. In particular, such edges will either be directed from the last successor of $v$ in $\pert^*(\mu_i)$ to a first successor of $v$ in $\pert^*(\mu_{i+1})$ ({\em right edges}) or from the last successor of $v$ in $\pert^*(\mu_{i+1})$ to a first successor of $v$ in $\pert^*(\mu_{i})$ ({\em left edges}). Observe that, in any augmentation $\pert^*(\mu)$ of $\pert(\mu)$ such that $\left< \pert(\mu),\pert^*(\mu) \right>$ is $s_\mu$-bitonic, for each pair of consecutive virtual edges $e_i$ and $e_{i+1}$ exiting $v$, either a left edge or the alternative right edge is introduced connecting a vertex in $\pert(\mu_i)$ with a vertex in $\pert(\mu_{i+1})$.

We assign a label in $\{L,R\}$ to some of the faces of $\skel(\mu)$ as follows. 
For each face $f$ of $\skel(\mu)$ incident to two consecutive virtual edges exiting $v$, we say that $v$ is the {\em source vertex} of $f$ if it is the source of the \st-graph induced by the edges incident to $f$, and $t_f$ is the {\em sink vertex} of $f$ if it is the sink of the \st-graph induced by the edges incident to $f$. Consider the two virtual edges $e_l=(v, v_l)$ and $e_r=(v,v_r)$ exiting $v$ and incident to $f$, where $e_l$ precedes $e_r$ in the clockwise order of the edges exiting $v$. 
If $v_l = t_f$ and $(s_f,v_l) \in \pert^*(\mu_l)$, we assign label $L$ to $f$.
If $v_r = t_f$ and $(s_f,v_r) \in \pert^*(\mu_r)$, we assign label $R$ to $f$.
% Observe that, some of the faces of $\skel(\mu)$ have not been assigned a label yet.
Observe that, in any augmentation $\pert^*(\mu)$ of $\pert(\mu)$ such that $\left< \pert(\mu),\pert^*(\mu) \right>$ is $s_\mu$-bitonic, faces with label $L$ (with label $R$) must be traversed by a left edge (resp. right edge). In fact, vertex $t_f$ is also the sink of the \st-graph induced by the edges incident to the face of $\pert(\mu)$ corresponding to $f$; also, the alternative edges with respect to those inserted would exit $t_f$ and hence would introduce a directed cycle in $\pert^*(\mu)$. We remark that, augmenting an unlabeled face with any of the two alternative edges does not introduce any directed cycles. This is due to the fact that there exists no directed path connecting an internal vertex in $\pert^*(\mu_l)$ with an internal vertex in $\pert^*(\mu_r)$. Hence, in the following we can assume that the obtained augmentation $\pert^*(\mu)$ of $\pert(\mu)$ is an acyclic \st-graph.

Based on the type of the children $\mu_1,\dots,\mu_k$ of $\mu$ and on the labeling of the faces of which $v$ is the source vertex, one of the following three claims applies.

\begin{claimx}
\label{cl:caseone}
If no child of $\mu$ corresponding to a virtual edge exiting $v$ is of {\texttt Type B} and if $v$ is not the source of two faces of $\skel(\mu)$ labeled $L$ and $R$, respectively, then $\pert(\mu)$ can be augmented  in such a way that $\mu$ is of {\texttt Type M}.
\end{claimx}

\begin{proof} Suppose that $v$ is not the source of any $R$-labeled face (resp., of any $L$-labeled face). For each $i=1,\dots,k$, we select the flip of ${\cal E}_i$ such that the last successor of $v$ in $\pert^*(\mu_i)$ lies on the left path of $\pert^*(\mu_i)$ (resp., on the right path of $\pert^*(\mu_i)$).
For each $i=1,\dots,k-1$, we add a left edge (resp., right edge) directed from the last successor of $v$ in $\pert^*(\mu_{i+1})$ (resp., in $\pert^*(\mu_{i})$) to the first successor of $v$ in $\pert^*(\mu_{i})$ (resp., in $\pert^*(\mu_{i+1})$). By Case~(\ref{obs:ext-vert-1}) of Observation~\ref{obs:ext-vert}, the introduced edges do not affect planarity. 
Finally, by the fact that all the nodes corresponding to the virtual edges exiting $v$ are of {\texttt Type M} and by the choice of the left and right edges, the obtained augmentation $\pert^*(\mu)$ of $\pert(\mu)$ is such that $\left< \pert(\mu),\pert^*(\mu) \right>$ is $v$-monotonic. 
\end{proof}

\begin{claimx}
\label{cl:casetwo}
If exactly one child $\mu_b$ of $\mu$ corresponding to a virtual edge exiting $v$ is of {\texttt Type B}, then  
$\pert(\mu)$ can not be augmented  in such a way that $\mu$ is of {\texttt Type M} while it can be augmented in such a way that $\mu$ is of {\texttt Type B} if and only if 
all the faces of $\skel(\mu)$ of which $v$ is the source vertex labeled $R$ (resp., labeled $L$) precede $e_b$ (resp., follow $e_b$) clockwise around $v$.
\end{claimx}

\begin{proof}
Clearly, in this case node $\mu$ cannot be of {\texttt Type M}.
First, observe that if vertex $v$ is the source of an $L$-labeled face $f_L$ of $\skel(\mu)$ that precedes the virtual edge $e_b$ clockwise around $v$, then node $\mu$ cannot be of {\texttt Type B} either. In fact,
in any augmentation $\pert^*(\mu)$ of $\pert(\mu)$ the subgraph of $\pert^*(\mu)$ induced by the successors of $v$ in $\pert(\mu)$ would contain the left edge traversing $f_L$ that points away from the apex of $v$ in $\pert^*(\mu_B)$. 
Analogously, observe that if vertex $v$ is the source of an $R$-labeled face $f_R$ of $\skel(\mu)$ that follows the virtual edge $e_b$ clockwise around $v$, then node $\mu$ cannot be of {\texttt Type B}.
Therefore, it remains to consider the case in which all the faces of $\skel(\mu)$ of which $v$ is the source vertex that are labeled $R$ (resp., labeled $L$) precede $e_B$ (resp., follow $e_b$) clockwise around $v$.
We can then augment $\pert(\mu)$ in such a way that $\left< \pert(\mu),\pert^*(\mu) \right>$ is strictly $v$-bitonic as follows.
For $i=1,\dots,b-1$, we select the flip of ${\cal E}_i$ such that the last successor of $v$ in $\pert^*(\mu_i)$ lies on the right path of $\pert^*(\mu_i)$; also, for $i=b+1,\dots,k$, we select the flip of ${\cal E}_i$ such that the last successor of $v$ in $\pert^*(\mu_i)$ lies on the left path of $\pert^*(\mu_i)$. 
For $i=1,\dots,b-1$, we add a right edge directed from the last successor of $v$ in $\pert^*(\mu_i)$ to a first successor of $v$ in $\pert^*(\mu_{i+1})$; also, for $i=b,\dots,k-1$, we add a left edge directed from the last successor of $v$ in $\pert^*(\mu_{i+1})$ to a first successor of $v$ in $\pert^*(\mu_{i})$. Clearly, the obtained augmentation of $\pert(\mu)$
is such that $\left< \pert(\mu),\pert^*(\mu) \right>$ is strictly $v$-bitonic and, by Observation~\ref{obs:child-orientation}, $\pert^*(\mu)$ is also planar.
\end{proof}

\begin{claimx}
\label{cl:casethree}
If no child of $\mu$ corresponding to a virtual edge exiting $v$ is of {\texttt Type B} and $v$ is the source vertex of at least one $R$-labeled face and of one $L$-labeled face, then  
$\pert(\mu)$ can not be augmented  in such a way that $\mu$ is of {\texttt Type M} while it can be augmented in such a way that $\mu$ is of {\texttt Type B} if and only if 
all the faces of $\skel(\mu)$ of which $v$ is the source vertex labeled $R$ precede the faces labeled $L$ clockwise around $v$.
\end{claimx}

\begin{proof}
First, observe that if vertex $v$ is the source of two faces $f_L$ and $f_R$ of $\skel(\mu)$ labeled $L$ and $R$, respectively, such that $f_L$ precedes $f_R$ clockwise around $v$, then there exists no
augmentation $\pert^*(\mu)$ of $\pert(\mu)$ such that $\left< \pert(\mu),\pert^*(\mu) \right>$ is $v$-bitonic. In fact, in any augmentation $\pert^*(\mu)$ of $\pert(\mu)$ the subgraph of $\pert^*(\mu)$ induced by the successors of $v$ in $\pert(\mu)$ would contain the left edge traversing $f_L$ followed by the right edge traversing $f_R$; clearly, this precludes a bitonic path.
Therefore, it only remains to consider the case in which all the faces labeled $R$ precede all the faces labeled $L$ in the clockwise order around $v$. 
Node $\mu$ cannot be of {\texttt Type M} as the existence of an apex vertex of $v$ is implied by the presence of both a left and a right edge.

We can then augment $\pert(\mu)$ in such a way that
$\left< \pert(\mu),\pert^*(\mu) \right>$ is strictly $v$-bitonic as follows.
Let $e_c$ be any virtual edge exiting $v$ such that all the $R$-labeled faces precede $e_c$ clockwise around $v$ and such that all the $L$-labeled faces follow $e_c$ clockwise around $v$.
We apply the same strategy as in the proof of Claim%$~\ref{cl:casetwo}
~\red{2}
 to select a flip for each embedding ${\cal E}_i$ and to introduce left and right edges to obtain $\pert^*(\mu)$, where $e_c$ has the role of $e_b$. Therefore, $\pert^*(\mu)$ is planar and $\left< \pert(\mu),\pert^*(\mu) \right>$ is strictly $v$-bitonic. Also, the apex of $v$ is the last successor of $v$ in $\pert^*(\mu_c)$. 
\end{proof}
{}

\subsection{Omitted Proofs of the Open Problems Section}\label{se:proofsopen}

\begin{figure}
  \centerline{\includegraphics{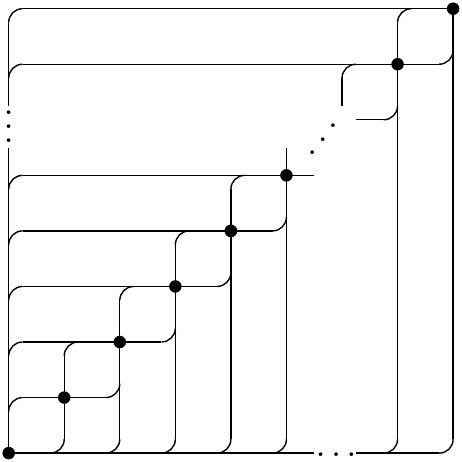}}
  \caption{\label{FIG:max}An L-planar graph with $n$ vertices and $4n-6$ edges.}
\end{figure}
\setcounter{lemma}{4}
\begin{lemma}
  A graph with $n$ vertices that admits a planar, upward-planar, or
  upward-rightward-planar L-drawing has at most $4n-6$, $3n-6$, or
  $2n-3$ edges and these bounds are tight.
\end{lemma}
\begin{proof} In the following let $n$ denote the number of vertices of the considered graph.
\begin{description} 
\item[planar:] Consider for each port of a vertex the furthest
  bend. Recall that the bend on any edge is the furthest bend of at
  least one of its end vertices. On the other hand each vertex has at
  most four furthest bends.  Thus there can be at most $4n$ edges.
  Consider now the outer face. The topmost (bottommost, rightmost,
  leftmost) vertex doesn't have a furthest bend at its top (bottom,
  right, left) port. Moreover in a maximal L-planar drawing there are
  at least two edges $e_1$ and $e_2$ on the outer face such that its
  bend is a furthest bend of both end vertices: Consider the
  bottommost vertex $v$. If $v$ is neither the leftmost nor the
  rightmost vertex, let $u_1$ and $u_2$ be the leftmost and rightmost
  vertex such that there is an edge $e_1=(u_1,v)$ and $e_2=(u_2,v)$,
  respectively. If $v$ is the leftmost (rightmost) vertex, let $u$ be
  the rightmost (leftmost) vertex such that there is an edge
  $e_1=(u,v)$ and let $w$ be the topmost vertex such that there is an
  edge $e_2=(v,w)$. This yields the $4n-6$ bound. Finally,
  Fig.~\ref{FIG:max} indicates a graph with $4n-6$ edges.
\item[upward-planar:] By Corollary~\ref{COR:undirected}, every maximal
  undirected graph oriented according to a bitonic st-ordering is a
  directed graph with $3n-6$ edges admitting an upward-planar
  L-drawing. Since upward-planar graphs must be acyclic, they cannot
  contain 2-cycles. Thus, there are at most $3n-6$ edges.
\item[upward-rightward-planar:]Each vertex has at most two furthest
  bends. The bottommost vertex has no furthest bend to the left, the
  rightmost vertex has no furthest bend to the top and in a maximal
  upward-rightward planar L-drawing there is at least one bend that is
  furthest for both end vertices. Hence, there are at most $2n-3$
  edges. Omitting all but the upward-rightward edges in
  Fig.~\ref{FIG:max} yields a graph with $2n-3$ edges.
\end{description}
\end{proof}

\end{document}